\documentclass[journal,twoside,draftclsnofoot,onecolumn]{IEEEtran} % 1 column 2 space

\usepackage[mathscr]{euscript}
\usepackage[pdftex]{graphicx}
\usepackage[sort,compress]{cite}
\usepackage{rotating}
\usepackage[disable]{todonotes}
\usepackage{comment}

\graphicspath{{./fig/}}

\usepackage{mathtools}
\usepackage{enumitem}
\usepackage{amssymb}
\usepackage{url}
\usepackage{afterpage}
\ifCLASSOPTIONcompsoc
  \usepackage[caption=false,font=normalsize,labelfont=sf,textfont=sf]{subfig}
\else
  \usepackage[caption=false,font=footnotesize]{subfig}
\fi

\usepackage{amsthm}
\usepackage{todonotes}

\newtheorem{theorem}{Theorem}
\newtheorem{props}{Proposition}

\theoremstyle{definition}
\newtheorem{defn}{Definition}

\newcommand{\sign}[1]{{\rm sign}\left(#1 \right)}

\newcommand{\Exp}{{\mathbb{E}}}

\newcommand{\dd}{{\rm d}}

\newcommand{\SEve}{\mathcal{S}_{\rm Eve}}
\newcommand{\SSteve}{\mathcal{S}_{\rm Steve}}
\newcommand{\Pro}{\mathbb{P}}
\newcommand{\asimeq}{\,{\stackrel{{\scriptscriptstyle{n \rightarrow \infty}}}{\displaystyle\simeq}}\,}
\newcommand{\ie}{\emph{i.e.}}
\newcommand{\eg}{\emph{e.g.}}

\renewenvironment{proof}[1][\proofname]{{\noindent \bfseries #1.}}{\qed\vspace{2pt} \mbox{}

}

\usepackage[yyyymmdd,hhmmss]{datetime}
\usepackage{units}
\usepackage{amsmath}

\newcommand{\mytitle}{On Known-Plaintext Attacks to a Compressed Sensing-Based Encryption: A Quantitative Analysis}
\begin{document}
\title{On Known-Plaintext Attacks to a Compressed Sensing-Based Encryption: A Quantitative Analysis}

\author{
\IEEEauthorblockN{
Valerio Cambareri, \IEEEmembership{Student Member, IEEE}, Mauro Mangia, \IEEEmembership{Member, IEEE}, \\ Fabio Pareschi, \IEEEmembership{Member, IEEE}, Riccardo Rovatti, \IEEEmembership{Fellow, IEEE}, Gianluca Setti, \IEEEmembership{Fellow, IEEE}}

% Copyright paragraph as requested.
\thanks{Copyright \copyright 2015 IEEE. Personal use of this material is permitted. However, permission to use this material for any other purposes must be obtained from the IEEE by sending a request to \url{pubs-permissions@ieee.org}.}
\thanks{V. Cambareri and R. Rovatti are with the Department of Electrical, Electronic and Information Engineering (DEI), University of Bologna, Italy (e-mail: \url{valerio.cambareri@unibo.it}, \url{riccardo.rovatti@unibo.it}).}%
\thanks{M. Mangia is with the Advanced Research Center on Electronic Systems (ARCES), University of Bologna, Italy (e-mail: \url{mmangia@arces.unibo.it}).}%
\thanks{F. Pareschi and G. Setti are with the Engineering Department in Ferrara (ENDIF), University of Ferrara, Italy (e-mail: \url{fabio.pareschi@unife.it}, \url{gianluca.setti@unife.it}).}%
}

\markboth{IEEE Transactions on Information Forensics and Security}{Cambareri \MakeLowercase{\emph{et al.}:} \textsc{\mytitle}}
\maketitle

\begin{abstract}
Despite the linearity of its encoding, compressed sensing may be used to provide a limited form of {data protection} when random encoding matrices are used to produce sets of low-dimensional measurements (ciphertexts). 
In this paper we quantify by theoretical means the resistance of the least complex form of this kind of encoding against known-plaintext attacks. For both standard compressed sensing with antipodal random matrices and recent multiclass encryption schemes based on it, we show how the number of candidate encoding matrices that match a typical plaintext-ciphertext pair is so large that the search for the true encoding matrix inconclusive. Such results on the practical ineffectiveness of known-plaintext attacks underlie the fact that even closely-related signal recovery under encoding matrix uncertainty is doomed to fail.

Practical attacks are then exemplified by applying compressed sensing with antipodal random matrices as {a multiclass} encryption scheme to signals such as images and electrocardiographic tracks, showing that the extracted information on the true encoding matrix from a plaintext-ciphertext pair leads to no significant signal recovery quality increase. This theoretical and empirical evidence clarifies that, although not perfectly secure, both standard compressed sensing and multiclass encryption schemes feature a noteworthy level of security against known-plaintext attacks, therefore increasing its appeal as a negligible-cost encryption method for resource-limited sensing applications.
\end{abstract}


\IEEEpeerreviewmaketitle

\section{Introduction}
\IEEEPARstart{T}{his} paper elaborates on the possibility of exploiting Compressed Sensing (CS) \cite{Donoho_2006, Candes_2008} {not only to reduce the resource requirements for signal acquisition, but also to protect the acquired data so that their information is hidden from unauthorised receivers}. A number of prior analyses \cite{rachlin2008secrecy,orsdemir2008security, cambareri2013twoclass, barcelo2014amplify,  cambareri2014low} show that, although the encoding performed by CS cannot be regarded as perfectly secure, practical encryption is still provided at a very limited cost, either at the analog-to-digital interface or immediately after it, in early digital-to-digital processing stages.

Such a lightweight encryption scheme may be particularly beneficial to acquisition systems within the framework of wireless sensor networks \cite{akyildiz2002survey} where large amounts of data are locally acquired by sensor nodes with extremely tight resource budgets, and afterwards transmitted to a remote node for further processing. When the security of these transmissions is an issue, low-resource techniques that help balancing the trade-off between encryption strength and computational cost may offer an attractive design alternative to the deployment of separate conventional encryption stages.

An encryption scheme based on CS leverages the fact that, in its framework, a high-dimensional signal is encoded by linear projection on a random subspace, thus producing a set of low-dimensional measurements.
These can be mapped back to the acquired signal only under prior assumptions on its {\em sparsity} \cite{candes2007sparsity} and a careful choice of random subspaces such as those defined by {antipodal} random {(also known as Bernoulli random \cite{candes2006near, haboba2012a})} encoding matrices. 
In addition, suitable {\em sparse signal recovery algorithms} \cite{Candes_2005, rangan2011generalized, van2011sparse} are required to decode the original signal. These must be applied with an exact knowledge of the subspace on which the signal was projected. In complete absence of this information the acquired signal is unrecoverable. Hence, this subspace may be {generated from} {a shared secret} between the transmitter and intended receivers that enables their high-quality signal recovery. 

If, on the other hand, the above subspace is only {\em partially} known, a low-quality version of the signal may be recovered from its measurements, with a degradation that increases gracefully with the amount of missing information on the projection subspace. 
By exploiting this effect, multiclass encryption schemes were devised \cite{cambareri2013twoclass, cambareri2014low} in which {\em high-class users} are able to decode high-quality information starting from a complete knowledge of the shared secret, while {\em lower-class users} only recover a low-quality approximation of the acquired signal starting from partial knowledge of the secret.
In order to take full advantage of this scheme, its security must be quantitatively assessed against potential cryptanalyses. 
The theoretical and empirical evidence provided in \cite{cambareri2014low} dealt with statistical attacks on the measurements produced by universal random encoding matrices \cite{candes2006near}.

In this paper we address the resistance of {an embodiment of} CS against Known-Plaintext Attacks (KPAs), \ie, in threatening situations where a malicious eavesdropper has gained access to an instance of the signal ({plaintext}) 
 and its corresponding random measurements ({ciphertext}),
 and from this information tries to infer the corresponding instance of {an antipodal} random encoding matrix.  
KPAs are more threatening than attacks solely based on observing the ciphertext. 
Yet, we will show how both simple and multiclass encryption based on CS exhibit a noteworthy level of resistance against this class of attacks due to the nature of the encoding.

The paper is organised as follows. In Section \ref{sec:atcdhs} we briefly review the fundamentals of CS and multiclass encryption in the two-class case, which distinguishes between {\em first-class receivers} authorised to reconstruct the signal with full quality and {\em second-class receivers} with reduced decoding quality.

Section \ref{sec:kpa} describes KPAs as delivered both by eavesdroppers and second-class receivers who aim at improving the quality of their signal recovery. 
There, it is shown that the expected number of candidate solutions matching a plaintext-ciphertext pair is enormous, thus implying that finding the true encoding matrix among such a huge solution set is practically infeasible. 
{To extend this analysis, we also attack the two-class encryption scheme by using recovery algorithms that compensate encoding matrix perturbations \cite{zhu2011sparsity, parker2011compressive} as suffered by a second-class receiver. Their performances are shown to be equal to a standard decoding algorithm \cite{rangan2011generalized} that does not attempt such compensation, \ie, that legitimately recovers the acquired signal at the prescribed quality level.}

In Section \ref{sec:ae} the previous KPAs are exemplified for electrocardiographic tracks (ECG) and images containing sensitive identification text. For all these cases we give empirical evidence on how, even in favourable attack conditions, the encoding matrices produced by KPAs perform poorly when trying to decode any further ciphertext. Theoretical and empirical evidence allows us to conclude that compressed sensing-based encryption, albeit not perfectly secure \cite{rachlin2008secrecy}, provides some security properties and defines a framework in which their violation is non-trivial. 
The Appendices report the proofs of the Propositions and Theorems given in Section \ref{sec:kpa}.

{\subsection{Relation to Prior Work}
\noindent To prove how CS and multiclass encryption provide a satisfying level of privacy even against informed attacks, this work addresses the problem of finding all the instances of {an antipodal} random encoding matrix that map a known plaintext to the corresponding ciphertext, when both quantities are deterministic and digitally represented. Our analysis hinges on the connection between linear encoding by {antipodal} random matrices, the subset-sum problem \cite{martello1990knapsack} and its expected number of solutions \cite{sasamoto2001statistical}. 
While the authors of \cite{rachlin2008secrecy} proved how CS lacks perfect secrecy in the Shannon sense \cite{Shannon_SecrecySystems_1949}, both \cite{rachlin2008secrecy} and \cite{orsdemir2008security} contrasted this with computational security evidence substantially based on brute-force attacks. Our improvement in the specific, yet practically important case of antipodal random encoding matrices is in that our analysis predicts how the expected number of candidate solutions to a KPA varies with the plaintext dimensionality and its digital representation. 

In addition, we evaluate specific attacks to multiclass encryption by CS in the case of lower-class users attempting to upgrade their recovery quality. To assess the resistance of this strategy against KPAs, we apply a similar theoretical analysis. Then, we extend the attacks to include sparse signal recovery under matrix uncertainty \cite{zhu2011sparsity, parker2011compressive} based on the idea that missing information \cite{loh2012high}, perturbations \cite{herman2010general, cambareri2015average} and basis mismatches \cite{chi2011sensitivity} could be partially compensated, although we verify that is not the case with the random perturbation entailed by multiclass encryption.
}

\section{Multiclass Encryption by Compressed Sensing}
\label{sec:atcdhs}

\subsection{A Brief Review of Compressed Sensing}
\label{prelims}
\noindent The encryption schemes we consider in this paper are based on Compressed Sensing (CS) \cite{Donoho_2006, Candes_2008}, a mathematical framework in which a signal represented by a vector $x\in\mathbb{R}^n$ is acquired by applying a linear, dimensionality-reducing transformation $A: \mathbb{R}^n \rightarrow \mathbb{R}^m$ (\ie, the {\em encoding matrix}) to generate a vector of {\em measurements} $y = A x, y \in\mathbb{R}^m, m < n$.
To enable the recovery of $x$ given $y$, CS leverages the fact that $x$ is known to be {\em sparse} in a proper basis $D$, \ie, for any instance of $x$ its representation is $x=D s$ where $s\in\mathbb{R}^n$ has a number of non-zero entries at most $k \ll n$. {The results presented in this paper are independent of $D$, which we consider an orthonormal basis for the sake of simplicity}. In addition, the encoding matrix $A$ must obey some information-preserving guarantees \cite{candes2008restricted, Donoho_2010} that we assume verified throughout this paper and essentially impose that $m = \mathcal{O}(k \log n)$. The most relevant fact here is that when $A$ is a typical realisation of a random matrix with independent and identically distributed (i.i.d.) entries following a subgaussian distribution \cite{baraniuk2008simple} we are reassured that signal recovery is possible regardless of the chosen basis $D$. In fact, some signal recovery algorithms exist for which guarantees can be given with very high probability \cite{Candes_2005} along with an ever-growing plethora of fast iterative methods capable of reconstructing $x$ starting from $y$, $A$ and $D$. An essential decoding scheme is the convex optimisation problem known as {\em basis pursuit with denoising},
\begin{equation}
\label{eq:bpdn}
\hat{x} = \mathop{\text{arg} \min}_{\xi \in \mathbb{R}^n}\;\; \left\lVert D^{-1} \xi \right\rVert_{1}\;\;\;\;
\text{s.t.}\;\;  \left\lVert A \xi-y\right\rVert_2 \le \omega\tag{BPDN}
\end{equation}
where the $\ell_1$-norm in the objective function promotes the sparsity of $\hat{x}$ with respect to $D$, while the $\ell_2$-norm constraint enforces its fidelity to the measurements up to a threshold $\omega \geq 0$ that accounts for noise sources.
In particular, we here concentrate on operators $A \in \{-1,1\}^{m \times n}$ that are realisations of {an antipodal} random matrix with i.i.d. entries and equiprobable symbols $\{-1,1\}$ \cite{candes2006near}; such matrices are known to verify the above guarantees, and are remarkably ($i$) simple, and therefore suitable to be generated, implemented and stored in digital devices ($ii$) random in nature, thus suggesting the possibility of exploiting such randomness to generate an encryption mechanism using the linear encoding scheme of CS. Due to their {limited set of possible symbols $\{-1, 1\}$}, {such antipodal random} matrices are more easily subject to cryptanalysis; for this reason, we tackle them as a baseline for {those defined by a larger set of symbols}.

\subsection{Security and Two-Class Encryption by Compressed Sensing}
\label{sec:twoclcs}
\subsubsection{A Security Perspective}
\label{sec:seccs}
the knowledge of $A$ is necessary in the recovery of $x$ from $y$, since any error in its entries reflects on the quality of the recovered signal \cite{herman2010general}. A number of security analyses leveraging this fundamental fact were introduced \cite{rachlin2008secrecy, orsdemir2008security, cambareri2014low} in which CS is regarded as a symmetric encryption scheme, where the {\em plaintext} $x$ is mapped to the {\em ciphertext} $y$ by means of the linear transformation operated by $A$, \ie, the {\em encryption algorithm}. {The ciphertext is then stored or transmitted, and its intended receivers may decrypt $x$ by knowing $y$, the sparsity basis $D$, and by having a prior agreement on the {\em encryption key} or {\em shared secret} that is necessary to reproduce $A$.}

{
The ideal requirement for a secure application of CS (as noted in \cite{drori_compressed_2008,rachlin2008secrecy}) is that any encoding matrix instance is used for {at most} one plaintext-ciphertext pair; this implies the use of a potentially {\em infinite} sequence of encoding matrices $\{A^{[t]}\}_{t \in \mathbb{N}}$. In violation of this {\em non-repeatability} hypothesis, each $A^{[t]}$ could be simply recovered by collecting $n$ linearly independent plaintext-ciphertext pairs related by it, \ie, by solving a linear system of equations with the $m n$ entries of $A^{[t]}$ as the unknowns.}

{In practice, the encoding matrices are obtained by algorithmic expansion of the shared secret, \eg, by using the key as the seed of a pseudo-random number generator (PRNG) which outputs a reproducible bitstream. Due to its deterministic and finite-state nature, this stream yields a {\em periodic} sequence of encoding matrices $\{A^{[t \mod{P}]}\}_{t \in \mathbb{N}}$ repeating with period $P$, where each $A^{[t]}$ is obtained by mapping $m n$ distinct bits to antipodal symbols. 

Thus, the non-repeatability hypothesis will be granted by a system-level choice of an encryption key and PRNG that makes $P$ large enough to exceed any reasonable observation time.}

{However, such pseudo-random bitstreams may themselves be vulnerable to cryptanalysis if a few of their bits are exposed. As a simple example of this threat, assume that the encoding matrices are generated by a maximal-length shift register sequence \cite[Chapter 4]{golomb2005signal}, for which a $\unit[B_{\rm key}]{bit}$ seed grants $P = \left\lfloor\frac{2^{B_{\rm key}}-1}{m n}\right\rfloor$. Regrettably, such a sequence is easily cryptanalysed from only $2 B_{\rm key}$ of its bits by the well-known Berlekamp-Massey algorithm \cite{massey1969shift-register}. 

Hence, a successful KPA that retrieves even part of an encoding matrix, \eg, one of its rows, may expose just enough information to reveal the key and therefore break a CS-based encryption. To contrast this type of threat, our analysis shows how KPAs are incapable of revealing missing information on the true encoding matrices, whose symbols remain undetermined. 
}

\subsubsection{Two-Class Encryption}
\label{sec:twoclcl}
{in an extended version of this encryption framework, \ie, {\em{two-class encryption}} by CS  \cite{cambareri2013twoclass,cambareri2014low}, we consider a first sequence of matrices $\{A^{(0), [t]}\}_{t \in \mathbb{N}}, A^{(0), [t]} \in \{-1, 1\}^{m \times n}$ obtained by pseudo-random expansion of a seed ${\rm Key}\left(A^{(0)}\right)$. In parallel, a sequence of index pair sets $\{C^{(0), [t]}\}_{t \in \mathbb{N}}, C^{(0), [t]} \subset\{0,\dots,m-1\}\times\{0,\dots,n-1\}$ is obtained by pseudo-random expansion of a seed ${\rm Key}\left(C^{(0)}\right)$. We then generate a second sequence of matrices $\{A^{(1), [t]}\}_{t \in \mathbb{N}}$ whose elements $A^{(1), [t]}$ are obtained by combining $A^{(0),[t]}, C^{(0),[t]}$ as
\begin{equation}
\label{signflip}
A^{(1), [t]}_{j,l}=\begin{cases}
A^{(0),[t]}_{j,l} & \text{if $(j,l)\not\in C^{(0),[t]}$}\\
-A^{(0),[t]}_{j,l} & \text{if $(j,l)\in C^{(0),[t]}$}\\
\end{cases}
\end{equation}
\noindent with $C^{(0),[t]}$ indicating which entries of $A^{(0),[t]}$ must be sign-flipped to obtain $A^{(1),[t]}$, that is then used to encode $x$ into $y$. 
Thus, we consider a cardinality $c$ for every $C^{(0),[t]}$, define $\eta=\nicefrac{c}{m n}$ the sign flipping density}, and let $A^{(0)}, A^{(1)}, C^{(0)}$ be generic, unique random matrix instances (that is, the matrix sequences will be implicitly considered from now on). 
Given any plaintext $x$, the corresponding ciphertext $y$ is produced as $y = A^{(1)} x$, $A^{(1)}$ being the {\em true encoding matrix}. Two-class encryption is then achieved by distributing ${\rm Key}\left(A^{(0)}\right)$ to all authorised receivers and ${\rm Key}\left(C^{(0)}\right)$ only to first-class receivers. In fact, when $y$ is communicated, receivers knowing both ${\rm Key}\left(A^{(0)}\right)$ and ${\rm Key}\left(C^{(0)}\right)$ are able to rebuild the corresponding $A^{(1)}$ used in the encoding and reconstruct  $x$ with full quality by solving BPDN with $\omega = 0$.

On the other hand, second-class receivers may only rebuild $A^{(0)}$ from their available information. For $0<\eta\ll 1$ such a matrix is an approximation of the corresponding $A^{(1)}$, thus allowing signal recovery with lower quality than that achieved by first-class receivers. Furthermore, any receiver not knowing ${\rm Key}\left(A^{(0)}\right)$ has no information on the encoding matrix and is consequently unable to recover $x$, which remains encrypted.

In \cite{cambareri2014low} we have characterised the effectiveness of this scheme by showing how eavesdroppers trying to compensate their ignorance of the key by means of straightforward statistical analysis of $y$ are presented with approximately Gaussian-distributed ciphertexts (converging with rate $\mathcal{O}(n^{-1})$). In addition, if $A^{(0)}$ is an antipodal  random matrix, the same can be said of $A^{(1)}$ since the statistics of its equiprobable symbols are unaltered by $C^{(0)}$ used to build the latter from the former. Hence, the ciphertext is statistically indistinguishable from the one that could be produced by encoding the same plaintext with $A^{(0)}$ instead of $A^{(1)}$, and second-class users will also be unable to exploit the statistical properties of $y$.
{
\subsection{Signal Models and Assumptions}
\noindent Since the attacks we present rely on {deterministic} knowledge of $x$ and $y$, we assume throughout the paper that both plaintexts and ciphertexts are represented by digital words. For simplicity, we let $x = \{x_l\}^{n-1}_{l=0}$ be such that $x_l\in\{-L,\dots,-1, 0, 1,\dots,L\}$ for some integer $L > 0$. 
Note that the number of bits representing the plaintext in this fashion is at least $B_x = \lceil\log_2 (2 L + 1) \rceil$, so we may assume $B_x$ is less than a few tens in typical embodiments (actually, $B_x \leq \unit[32]{bit}$ {in typical signal processing applications}). Consequently, the ciphertext will be represented by $\{y_l\}^{m-1}_{l=0}$, where each $y_l$ is quantised with {$B_y = \unit[{B_x + \lceil \log_2 n \rceil}]{bit}$ that} avoid any information loss. 
\section{Known-Plaintext Attacks}
\label{sec:kpa}
\begin{figure}[t]
\centering
\includegraphics[width=3.3in]{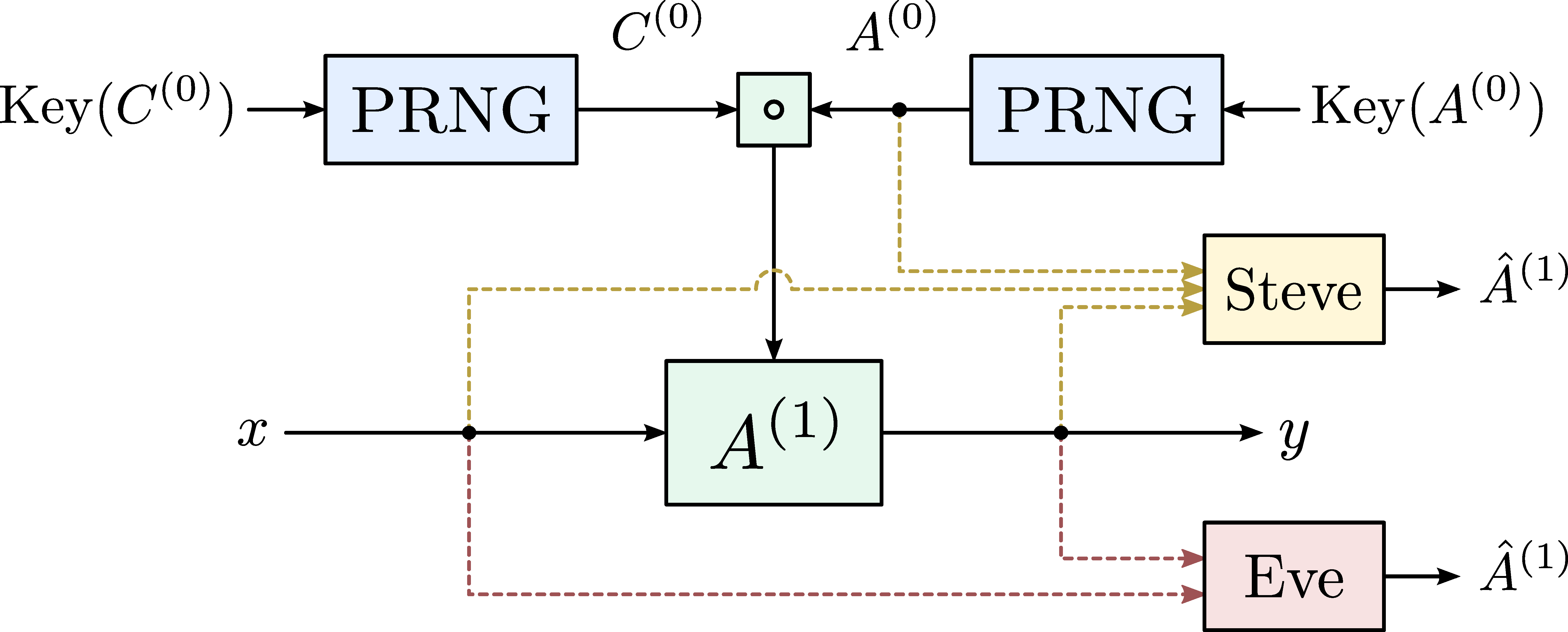}
\caption{\label{schemascemo}A two-class encryption scheme and the known-plaintext attacks being analysed from an eavesdropper (Eve) and a second-class user (Steve).}
\end{figure}
\noindent In view of quantifying the resistance of this scheme to threatening cryptanalyses, we now consider situations in which an attacker gains access to a given, exact value of the plaintext $x$ corresponding to a ciphertext $y$.  Based on this knowledge, the attacker aims at computing the true encoding $A^{(1)}$ such that $y = A^{(1)} x$. In the following we will consider a KPA by assuming that only one $(x, y)$ pair is known for a certain $A^{(1)}$, consistently with the hypothesis that $A^{(1)}$ is never reused in the encoding (as detailed in Section \ref{sec:seccs}). 
This type of attack gives rise to different strategies (see Fig. \ref{schemascemo}) whether the attacker knows nothing except the $(x,y)$ pair (a pure eavesdropper, Eve) or it is a second-class receiver knowing also the partially correct encoding $A^{(0)}$ and attempting to complete its knowledge of $A^{(1)}$ (we will call this malicious second-class user Steve and its KPA a \emph{class-upgrade}). 

{For the sake of simplicity, both KPAs are here characterised on a single row\footnote{We denote with $A_j$ the $j$-th row of a matrix $A$.} of $A^{(1)}$, while a complete KPA will entail $m$ of such attacks. Furthermore, we note that the analysis is carried out in full compliance with Kerckhoffs's principle \cite{kerckhoffs1883la}, \ie, the only information that the attackers are missing is their respective part of the encryption key, while any other detail on the sparsity basis, as well as two-class encryption specifications is here regarded as known.}

\subsection{Eavesdropper's Known-Plaintext Attack}
\label{evekpa}
\noindent Given a plaintext $x$ and the corresponding ciphertext $y = A^{(1)} x$ we now assume the perspective of Eve and attempt to recover $A^{(1)}_j$ with a set of antipodal symbols $\hat{A}^{(1)}_j = \{\hat{A}^{(1)}_{j,l}\}_{l=0}^{n-1}$ such that
\begin{equation}
\label{eq:pm}
y_j = \sum_{l=0}^{n-1} \hat{A}^{(1)}_{j,l} x_l
\end{equation}
Moreover, to favour the attacker\footnote{If any $x_l = 0$ each corresponding summand would give no contribution to the sum \eqref{eq:pm}, thus making $\hat{A}^{(1)}_{j,l}$ an undetermined variable in the attack.} we assume all $x_l \neq 0$. We now introduce a combinatorial optimisation problem at the core of the analysed KPAs.

\begin{defn}[Subset-Sum Problem]
Let $\{u_l\}^{n-1}_{l=0}, u_l \in \{1, \ldots, L\} \subset \mathbb{N}_+$ and $\upsilon \in \mathbb{N}_+$. We define \emph{subset-sum problem} (SSP) \cite[Chap. 4]{martello1990knapsack} the problem of assigning $n$ binary variables $b_l\in\{0,1\}$, $l=0,\dots,n-1$ such that
\begin{equation}
\label{ssCS}
\upsilon = \sum_{l=0}^{n-1} b_l u_l
\end{equation}
We define \emph{solution} any $\{b_l\}^{n-1}_{l=0}$ verifying \eqref{ssCS}.
With the above definitions, the  {\em density} of this problem is defined as \cite{lagarias1985solving}
\begin{equation}
\label{density}
\delta(n, L) = \frac{n}{\log_2 L}
\end{equation}

\end{defn}
Although in general a SSP is NP-complete, not all of its instances are equally {\em hard}. In fact, it is known that {\em high-density} instances (\ie, $\delta(n,L) > 1$) have {plenty} of solutions found or approximated by, \eg, dynamic programming, whereas {\em low-density} instances are {typically hard}, although for special cases polynomial-time algorithms have been found \cite{lagarias1985solving}. Moreover, such low-density hard SSP instances have been used in cryptography to develop the family of {public-key knapsack cryptosystems} \cite{merkle1978hiding, chor1988knapsack} although most have been broken with polynomial-time algorithms \cite{odlyzko1990rise}.
\begin{props}[Eve's KPA]
\label{prop:SS1}
The KPA to $A_j^{(1)}$ given $(x, y)$ is equivalent to a SSP where each $u_l=|x_l|$, the variables $b_l = \frac{1}{2}(\sign{x_l} \hat{A}^{(1)}_{j,l}  + 1)$ and the sum $\upsilon = \frac{1}{2}\left(y_j + \sum^{n-1}_{l=0} |x_l|\right)$. This SSP has a {\em true solution} $\{\bar{b}_l\}_{l=0}^{n-1}$ that is mapped to the row $A^{(1)}_j$, and other {\em candidate solutions} that verify \eqref{ssCS} but correspond to matrix rows $\hat{A}^{(1)}_j \neq A^{(1)}_j$.
\end{props}
This mapping is explained in Appendix \ref{appSS}, and we define $(x,y, A^{(1)}_j)$ a {\em problem instance}.  In our case we see that the density \eqref{density} is high since $n$ is large and $\log_2 L$ is fixed by the digital representation of $x$ (\eg, so that $B_x \leq 64$). We are therefore operating in a region {in which a solution of the SSP \eqref{ssCS} is typically found in polynomial time}. In fact, the resistance of the analysed embodiment of CS against KPAs is not due to the hardness of the corresponding SSP but, as we show below, to the huge number of candidate solutions as $n$ increases, among which an attacker should find the only true solution to guess a single row of $A^{(1)}$. Since no {\em a priori} criterion exists to select them, we consider them {\em indistinguishable}.
The next Theorem\footnote{{$\asimeq$ denotes asymptotic equality as $n \rightarrow \infty$.}}
 calculates the expected number of candidate solutions to Eve's KPA by applying the theory developed in \cite{sasamoto2001statistical}.
\begin{figure}[t]
	\centering
	\includegraphics[width=3.4in]{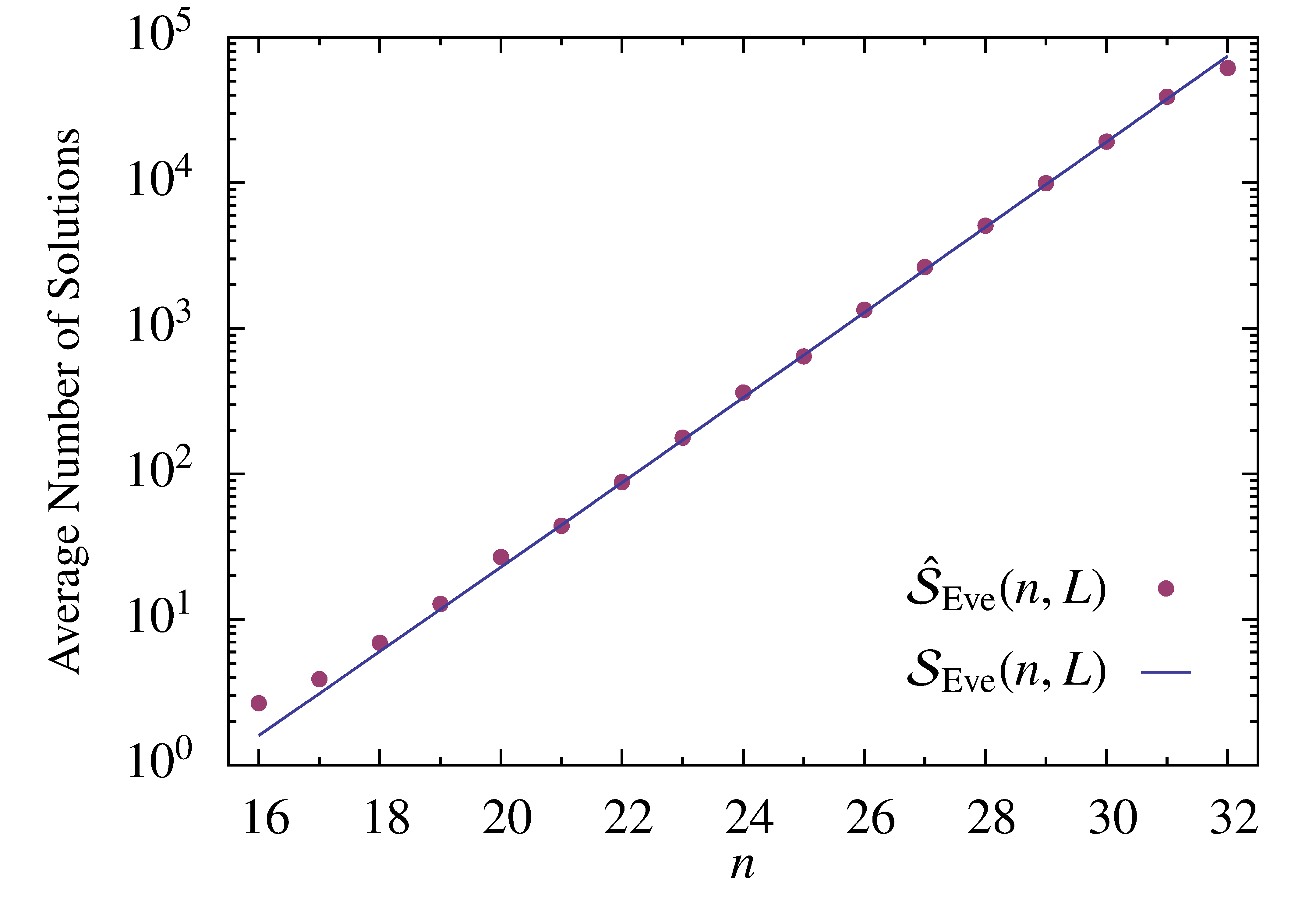}
	\caption{Sample average of the number of solutions for Eve's KPA compared to the theoretical value of \eqref{nosolSS} for $L = 10^4$.}
	\label{fig:enumsolKPeve}
\end{figure}

\begin{theorem}[Expected number of solutions for Eve's KPA]
\label{ExpSolSS}
For large $n$, the expected number of candidate solutions of the KPA in Proposition \ref{prop:SS1}, {in which ($i$) all the coefficients $\{u_l\}^{n-1}_{l=0}$ are i.i.d. uniformly drawn from $\{1, \ldots, L\}$, and ($ii$) the true solution $\{\bar{b}_l\}_{l=0}^{n-1}$ is drawn with equiprobable and independent binary values}, is
\begin{equation}
\SEve(n, L) \asimeq \frac{2^n}{L} \sqrt{\frac{3}{\pi n}}
\label{nosolSS}
\end{equation}
\end{theorem}

The proof of Theorem \ref{ExpSolSS} is given in Appendix \ref{appSS}. This result (as well as the whole statistical mechanics framework from which it is derived) gives no hint on how much \eqref{nosolSS} is representative of finite-$n$ behaviours. To compensate for that, {we here enumerate by means of the binary programming solver in CPLEX \cite{cplex} all the solutions to several small-$n$ problem instances of Proposition \ref{prop:SS1} and verify that, even non-asymptotically, the expression \eqref{nosolSS} can be used to effectively estimate the expected number of candidate solutions to Eve's KPA}. Such numerical evidence is reported in Fig. \ref{fig:enumsolKPeve}, where the sample average of the number of solutions $\hat{\mathcal{S}}_\text{\rm Eve}(n, L)$ to $50$ {randomly generated} problem instances with $L=10^4$ and $n = 16, \ldots, 32$ is plotted and compared with \eqref{nosolSS}. 

The remarkable matching observed therein allows us to estimate, for example, that a KPA to the encoding of a grayscale image of $n = \unit[64 \times 64]{pixel}$ quantised with $B_x = \unit[8]{bit}$ (unsigned, \ie, $L = 128, n = 4096$) would have to discriminate on the average between $1.25 \cdot 10^{1229}$ equally good candidate solutions for each of the rows of the encoding matrix. This number is not far from the total possible rows, $2^{4096}= 1.04 \cdot 10^{1233}$. Hence, any attacker using this strategy is faced with a deluge of candidate solutions, from which it would choose one presumed to be {exact} to attempt a guess on a single row of $A^{(1)}$.

A legitimate concern when the attacker is presented with such a set of solutions is that most of them could be good approximations of the true encoding matrix row $A^{(1)}_j$. To see whether this is the case, we quantify the difference between $A^{(1)}_j$ and the corresponding candidates $\hat{A}^{(1)}_j$ resulting from a KPA in terms of their Hamming distance, \ie, as the number of entries in which they differ. 

\begin{theorem}[Expected number of solutions for Eve's KPA at Hamming distance $h$ from the true one]
\label{ExpSolSSHam}
The expected number of candidate solutions {at Hamming distance $h$ from the true solution} of the KPA in Proposition \ref{prop:SS1}, {in which ($i$) all the coefficients $\{u_l\}^{n-1}_{l=0}$ are i.i.d. uniformly drawn from $\{1, \ldots, L\}$, ($ii$) the true solution $\{\bar{b}_l\}_{l=0}^{n-1}$ {is drawn with} equiprobable and independent binary values,} is
\begin{equation} 
\label{eq:hX1}
\SEve^{(h)}(n,L)=\binom{n}{h}\frac{P_h(L)}{2^h L^h}
\end{equation}
where $P_h(L)$ is a polynomial in $L$ whose coefficients are reported in Table \ref{tab:PhL} for $h=2,\dots,15$.
\end{theorem}

The proof of this Theorem and the derivation of Table \ref{tab:PhL} are reported in Appendix \ref{HDsols}. As before, we collect some empirical evidence that the expression \eqref{eq:hX1} correctly anticipates the expected number of solutions at a given Hamming distance from the true one, noting that Theorem \ref{ExpSolSSHam} holds for finite $n$. 
Figure \ref{fig:enumsolKPeveHam} reports for $n=\{21,23,\ldots,31\}$ the sample average, over {the same $50$ problem instances generated in the experimental evaluation of \eqref{nosolSS}}, of the number of solutions to Eve's KPA whose Hamming distance from the true one is a given value $h = \{2, \ldots, 15\}$.  {This sample average is} compared against the value predicted by \eqref{eq:hX1} with the polynomial coefficients in Table \ref{tab:PhL}. The remarkable matching we observe allows us to estimate that, resuming the case of a grayscale image with $n=4096, L=128$, only $1.95 \cdot 10^{41}$ candidate solutions out of the average $1.25 \cdot 10^{1229}$ are expected to have a Hamming distance $h \leq 16$, while $6.33 \cdot 10^{76}$ attain a Hamming distance $h \leq 32$. Since these results apply to each row of the matrix being inferred, this indicates how the chance that a randomly chosen candidate solution is close to the true one is negligible.

{
Under {\em repeated} threat of Eve's KPA, a system-level perspective would impose a change of encryption key (\ie, of encoding matrix sequence) whenever the probability of failure of repeated KPAs, $p_{\rm fail} $, drops below a desired security level $\zeta \in (0,1)$, \ie, at any time $p_{\rm fail} \geq \zeta$. 
Some insight on the {\em encryption key lifetime} $T$ that guarantees this is then obtained by modelling the repeated KPAs as i.i.d. Bernoulli trials, each leading to a successful choice of the true solution with a probability that can be estimated with $\SEve(n,L)^{-1}$ in case of Eve's KPA. With this $p_{\rm fail} = \Pro\left[T \ \text{KPA fail}\right] = (1-\SEve(n, L)^{-1})^T$, so we may choose the key lifetime as $T \leq {\left[\log\left(1-\SEve(n, L)^{-1}\right)\right]^{-1}} \, {\log \zeta}$ to ensure the security level set by $\zeta$. Thus, we measure the key lifetime $T$ in {\em attack opportunities} for Eve; however, since $\SEve(n,L)$ is typically huge, the resulting $T$ is also very large. As an example, by plugging $n = 4096, L = 128$ in \eqref{nosolSS} and assuming $\zeta = 0.9999$, we obtain a key lifetime equivalent to at most $T = 1.25 \cdot 10^{1225}$ attack opportunities.
}

\begin{figure}[t]
	\centering
	\vspace{.3pt}
		\includegraphics[width=3.4in]{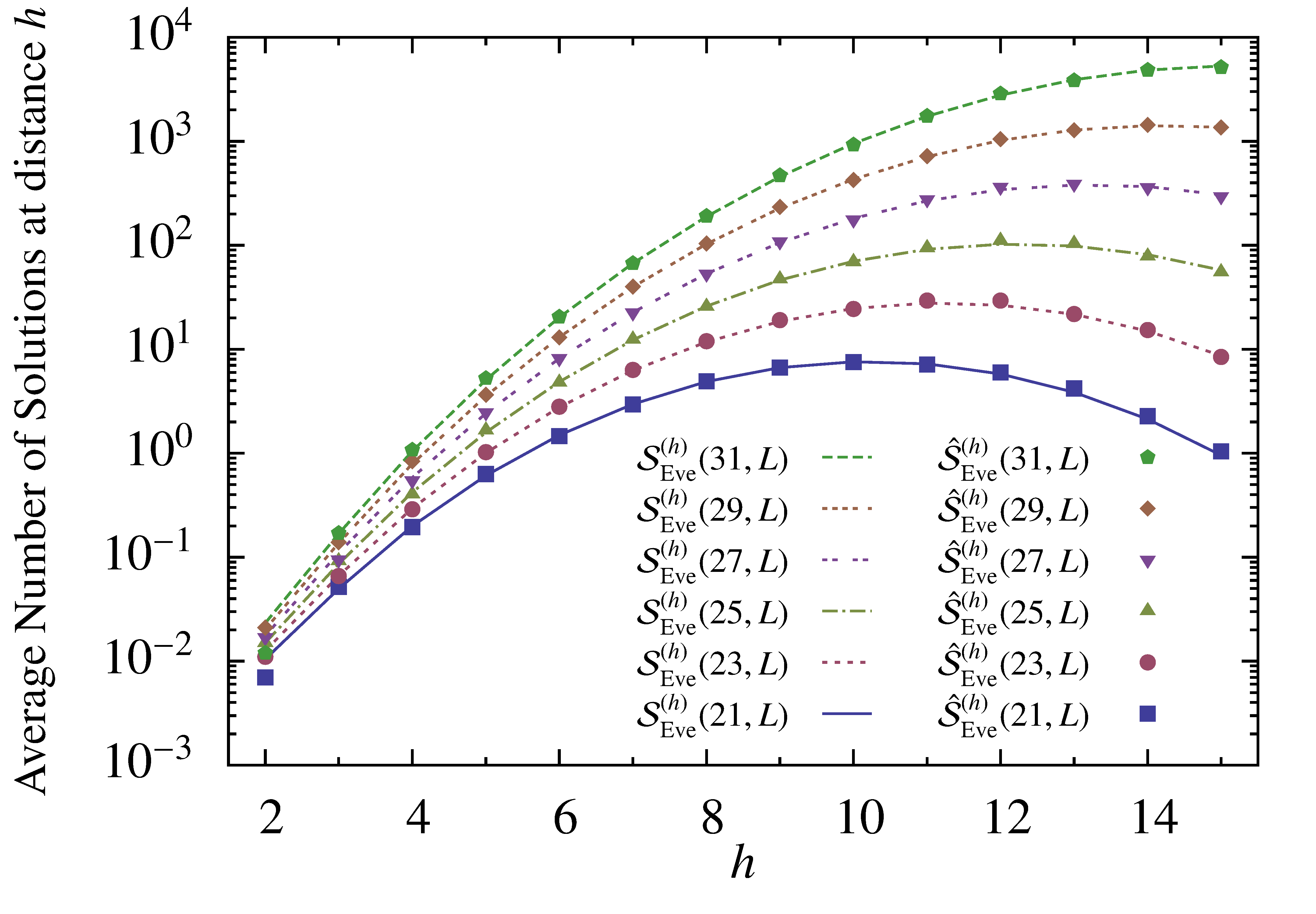}
	\caption{Sample average of the number of solutions for Eve's KPA at Hamming distance $h$ from the true one, compared to the theoretical value of \eqref{eq:hX1} for $L = 10^4$ and $n=21,23,\ldots,31$.}
	\label{fig:enumsolKPeveHam}
\end{figure}

\begin{table*}
\centering
% \resizebox{1.90\columnwidth}{!}{ % two col
\resizebox{\columnwidth}{!}{ % one col
$
\renewcommand{\arraystretch}{1.3}
\begin{array}{r||rrrrrrrrrrrrrr}
h & p_{1}^h & p_{2}^h & p_{3}^h & p_{4}^h & p_{5}^h & p_{6}^h & p_{7}^h & p_{8}^h & p_{9}^h & p_{10}^h & p_{11}^h & p_{12}^h & p_{13}^h & p_{14}^h \\
\hline
2 & 2 &&&&&&&&&&&&\\
\\
3 & -3 & 3 &&&&&&&&&&&\\
\\
4 & \dfrac{14}{3} & -4 & \dfrac{16}{3} &&&&&&&&&&\\
\\
5 & -\dfrac{15}{2} & \dfrac{65}{12} & -\dfrac{15}{2} & \dfrac{115}{12} &&&&&&&&&\\
\\
6 & \dfrac{62}{5} & -\dfrac{15}{2} & 11 & -\dfrac{27}{2} & \dfrac{88}{5} &&&&&&&&\\
\\
7 & -21 & \dfrac{959}{90} & -\dfrac{203}{12} & \dfrac{707}{36} & -\dfrac{301}{12} & \dfrac{5887}{180} &&&&&&&\\
\\
8 & \dfrac{254}{7} & -\dfrac{140}{9} & \dfrac{1226}{45} & -\dfrac{266}{9} & \dfrac{334}{9} & -\dfrac{422}{9} & \dfrac{19328}{315} &&&&&&\\
\\
9 & -\dfrac{255}{4} & \dfrac{2613}{112} & -\dfrac{731}{16} & \dfrac{14701}{320} & -\dfrac{457}{8} & \dfrac{2233}{32} & -\dfrac{1415}{16} & \dfrac{259723}{2240} &&&&&\\
\\
10 & \dfrac{1022}{9} & -\dfrac{2585}{72} & \dfrac{359105}{4536} & -\dfrac{7055}{96} & \dfrac{9869}{108} & -\dfrac{1725}{16} & \dfrac{28625}{216} & -\dfrac{48325}{288} & \dfrac{124952}{567} &&&&\\
\\
11 & -\dfrac{1023}{5} & \dfrac{16973}{300} & -\dfrac{60775}{432} & \dfrac{5463953}{45360} & -\dfrac{435941}{2880} & \dfrac{7449761}{43200} & -\dfrac{19811}{96} & \dfrac{1091629}{4320} & -\dfrac{2764663}{8640} & \dfrac{381773117}{907200} &&&\\
\\
12 & \dfrac{4094}{11} & -\dfrac{2277}{25} & \dfrac{687791}{2700} & -\dfrac{72523}{360} & \dfrac{3907067}{15120} & -\dfrac{341143}{1200} & \dfrac{599327}{1800} & -\dfrac{7909}{20} & \dfrac{1045349}{2160} & -\dfrac{2205833}{3600} & \dfrac{41931328}{51975} &&\\
\\
13 & -\dfrac{1365}{2} & \dfrac{591721}{3960} & -\dfrac{2020421}{4320} & \dfrac{44385419}{129600} & -\dfrac{7815847}{17280} & \dfrac{116257063}{241920} & -\dfrac{3192163}{5760} & \dfrac{110721221}{172800} & -\dfrac{13148473}{17280} & \dfrac{19285357}{20736} & -\dfrac{20345507}{17280} & \dfrac{20646903199}{13305600} &\\
\\
14 & \dfrac{16382}{13} & -\dfrac{44863}{180} & \dfrac{34353347}{39600} & -\dfrac{38237381}{64800} & \dfrac{1292711}{1600} & -\dfrac{42972293}{51840} & \dfrac{122732801}{129600} & -\dfrac{92420419}{86400} & \dfrac{53508931}{43200} & -\dfrac{76095383}{51840} & \dfrac{77441609}{43200} & -\dfrac{588168119}{259200} & \dfrac{866732192}{289575} \\
\\
15 & -\dfrac{16383}{7} & \dfrac{1074679}{2548} & -\dfrac{583763}{360} & \dfrac{113982839}{110880} & -\dfrac{12673507}{8640} & \dfrac{58584511}{40320} & -\dfrac{400088153}{241920} & \dfrac{1033251187}{564480} & -\dfrac{23927713}{11520} & \dfrac{193398181}{80640} & -\dfrac{98109773}{34560} & \dfrac{279340567}{80640} & -\dfrac{1060693411}{241920} & \dfrac{467168310097}{80720640} \\
\\
\end{array}
$
}
\vfill
\caption{\label{tab:PhL}Table of coefficients of the polynomials $P_h(L)=\sum_{j=1}^{h-1}p^h_j L^j$ in \eqref{eq:hX1} for $h=2,\dots,15$.}
\end{table*}

\subsection{Class-Upgrade Known-Plaintext Attack}
\label{cukpa}
\noindent A known-plaintext attack may also be attempted by Steve, a second-class receiver aiming to improve its signal recovery performances with the intent of reaching the same quality of a first-class receiver. In this KPA, a partially correct encoding matrix $A^{(0)}$ that differs from $A^{(1)}$ in $c$ entries is also known in addition to $x$ and $y$. With this prior, Steve may compute $\varepsilon=y - A^{(0)} x = \Delta A x$ where $\Delta A=A^{(1)} - A^{(0)}$ {here is an unknown matrix} with ternary entries in $\{-2, 0, 2\}$. Hence, Steve performs a KPA by searching for a set of ternary symbols $\{\Delta A_{j,l}\}_{l=0}^{n-1}$ such that
\begin{equation}
\label{CSkpsteve}
\varepsilon_j = \sum^{n-1}_{l = 0} \Delta A_{j, l} x_l 
\end{equation}
\noindent of which it is known that $\Delta A_{j,l}\neq 0$ only in $c$ cases. Moreover, to ease the solution of this problem and make it row-wise separable, we assume that Steve has access to an even more accurate information, \ie, the exact number $c_j$ of non-zero entries for each row $\Delta A_j$ or equivalently the number of sign flips mapping $A^{(0)}_j$ into the corresponding $A^{(1)}_j$ (clearly, the total number of non-zero entries in $\Delta A$ is $c = \sum_{j=0}^{m-1}c_j$). By assuming this, we may prove the equivalence between Steve's KPA to each row of $A^{(1)}$ and a slightly adjusted SSP.
\begin{defn}[$\gamma$-cardinality Subset-Sum Problem]
Let $\{u_l\}^{n-1}_{l=0}, u_l \in \{1, \ldots, Q\} \subset \mathbb{N}_+$, $\gamma \in \{1,\ldots,n\} \subset \mathbb{N}_+$ {and $\upsilon \in \mathbb{N}_+$}. We define $\gamma$-\emph{cardinality subset-sum problem} ($\gamma$-SSP) the problem of assigning $n$ binary variables $b_l\in\{0,1\}$, $l=0,\dots,n-1$ such that
\begin{eqnarray}
\label{sssteveCS}
\upsilon & = & \sum_{l=0}^{n-1} b_l u_l\\
\label{eq:cconstr}
\gamma & = & \sum_{l=0}^{n-1} b_l
\end{eqnarray}
We define \emph{solution} any $\{b_l\}^{n-1}_{l=0}$ verifying \eqref{sssteveCS} and \eqref{eq:cconstr}.
\end{defn}
\begin{props}[Steve's KPA]
\label{propsteveKPA}
The KPA to $A^{(1)}_j$ given $(x, y, A^{(0)}, c_j)$, is equivalent to a $\gamma$-SSP where $\gamma=c_j$, $Q = 2 L$, $\upsilon=\frac{1}{2}\varepsilon_j + L \, c_j$, $u_l=-A^{(0)}_{j,l} x_l + L$ and $b_l = \frac{1}{2}\left(1-\textstyle\frac{\hat{A}^{(1)}_{j, l}}{A^{(0)}_{j,l}}\right)$. This SSP has a {\em true solution} $\{\bar{b}_l\}_{l=0}^{n-1}$ that is mapped to the row $A^{(1)}_j$, and other {\em candidate solutions} that verify \eqref{sssteveCS} and \eqref{eq:cconstr} but correspond to matrix rows $\hat{A}^{(1)}_j \neq A^{(1)}_j$.
\end{props}
{The derivation of Proposition \ref{propsteveKPA} is reported in Appendix \ref{appSScon}. We define $(x,y, A^{(0)}_j, A^{(1)}_j)$ a {\em problem instance}.} 
In the following, we will denote with $r = \nicefrac{c_j}{n}$ the row-density of perturbations. Since in \cite{sasamoto2001statistical} the $\gamma$-cardinality SSP case is obtained as an extension of the results on the unconstrained SSP, we obtain the following Theorem.
\begin{theorem}[Expected number of solutions for Steve's KPA]
\label{ExpSolSSCon}
For large $n$, the expected number of candidate solutions of the KPA in Proposition \ref{propsteveKPA}, {in which ($i$) all the coefficients $\{u_l\}^{n-1}_{l=0}$ are i.i.d. uniformly drawn from $\{1, \ldots, 2 L\}$, and ($ii$) the true solution $\{\bar{b}_l\}_{l=0}^{n-1}$ is drawn with equiprobable independent binary values}, is
\begin{equation}
\SSteve(n, L, r) \asimeq \sqrt{
\frac{3}{2}
} \frac{
r^{-1-n r}\left(1 - r \right)^{-1-n(1-r)}
}
{
2 \pi n L
}
\label{nosolSSCon}
\end{equation}
\end{theorem}
\begin{figure*}[!t]
	\centering
	\null\hfill
	\subfloat[]{\centering\includegraphics[height=140pt]{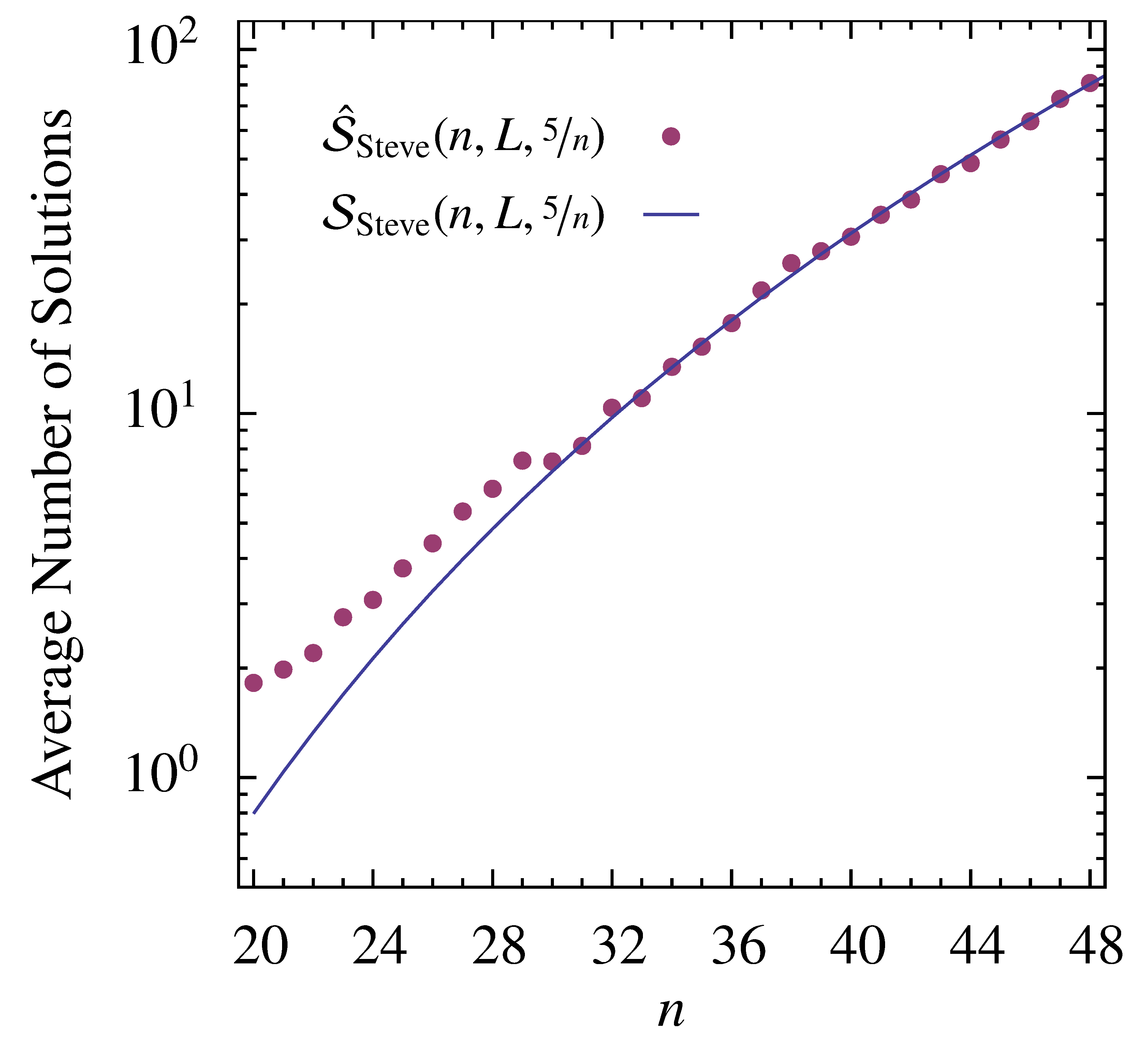}}
	\hfill
	\subfloat[]{\centering\includegraphics[height=140pt]{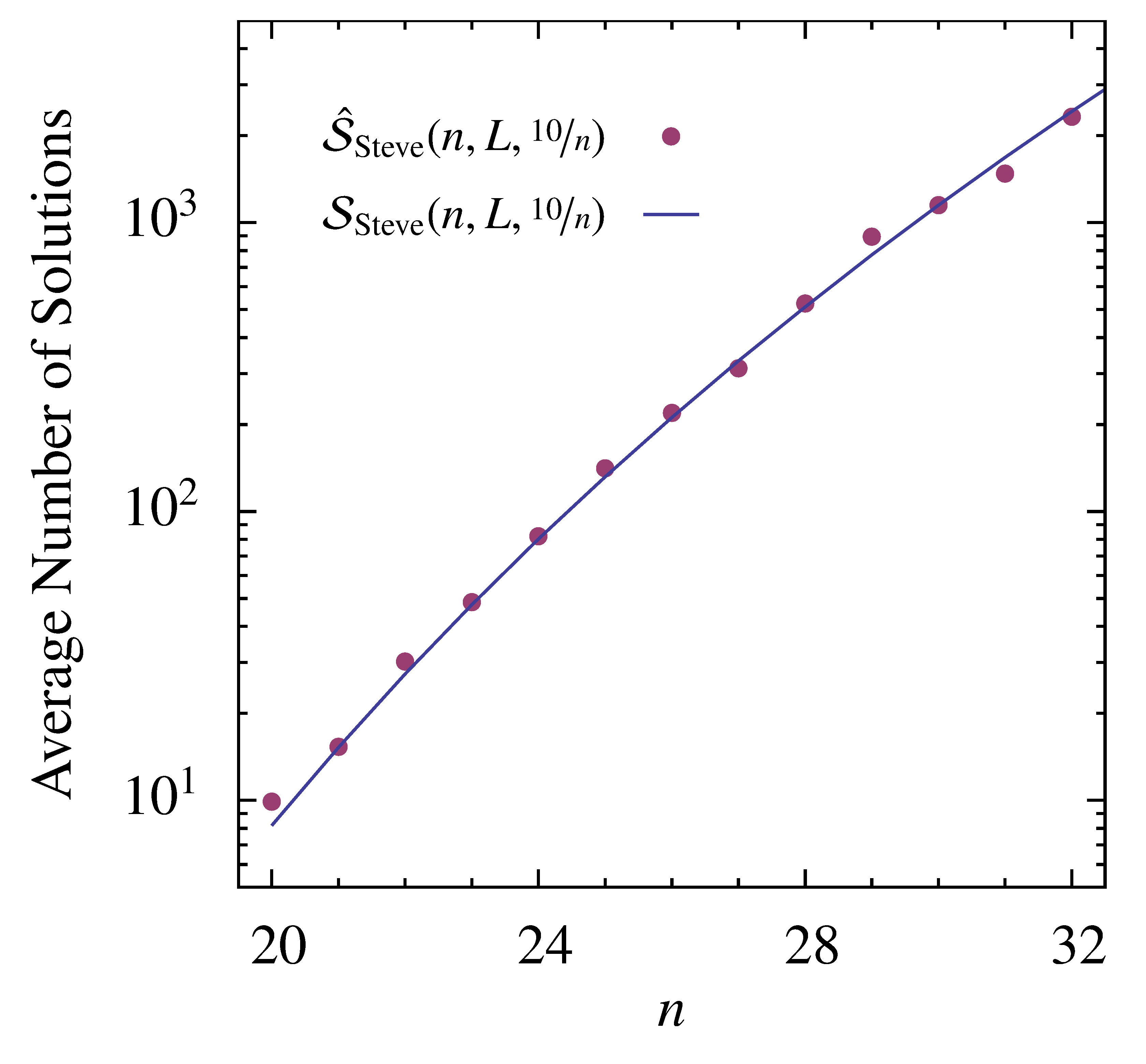}}
	\hfill
	\subfloat[]{\centering\includegraphics[height=140pt]{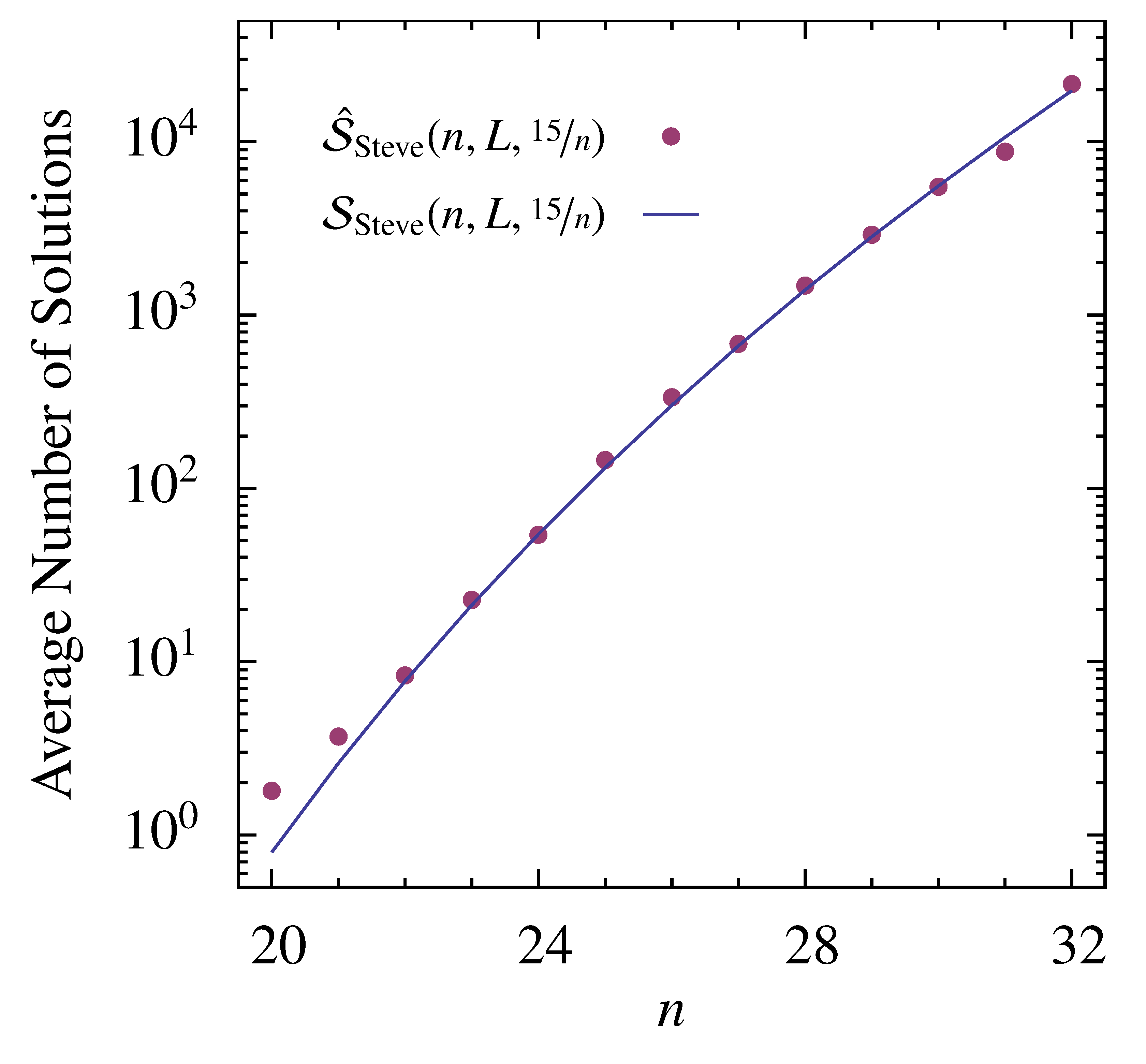}}
	\hfill\null
	\caption{Sample average of the number of solutions for Steve's KPA compared to the theoretical value of \eqref{nosolSSCon} for $L = 5 \cdot 10^3$ with row-density of perturbations $r = \nicefrac{5}{n}, \nicefrac{10}{n}, \nicefrac{15}{n}$.}
	\label{fig:enumsolKPsteve}
\end{figure*}
The proof of Theorem \ref{ExpSolSSCon} is reported in Appendix \ref{appSScon}. The number of candidate solutions found by Steve's KPA is by many orders of magnitude smaller than Eve's KPA, the reason being that Steve requires much less information to achieve complete knowledge of the true encoding $A^{(1)}$. 
In order to provide numerical evidence, {we find all the solutions to Steve's KPA by means of the binary programming solver in CPLEX} on a set of 50 {randomly generated} problem instances for $L=5 \cdot 10^3$, a row-density of perturbations $r = \nicefrac{5}{n}, \nicefrac{10}{n}, \nicefrac{15}{n}$ {and $n = 20, \ldots, 32$ (except for $r = \nicefrac{5}{n}$, whose solution enumeration is still computationally feasible up to $n = 48$)}. 
The sample average of the number of solutions, $\hat{\mathcal{S}}_\text{\rm Steve}(n, L, r)$, is reported in Fig. \ref{fig:enumsolKPsteve} and well predicted by the theoretical value in \eqref{nosolSSCon}; note that this approximation is increasingly accurate for large $n$. 
Moreover, by resuming the previous example our $n = \unit[64 \times 64]{pixel}$ grayscale image quantised at $B_x = \unit[8]{bit}$ and encoded with two-class CS using $\Delta A$ with $r = 0.03$ will have on the average $6.25 \cdot 10^{234}$ candidate solutions of indistinguishable quality. 

{
In terms of encryption key lifetime, leveraging the same considerations of Section \ref{evekpa} and simply replacing $\SEve(n,L)$ with $\SSteve(n, L, r)$ yields the key lifetimes $T$ with respect to class-upgrade attacks; as an example, plugging $n = 4096, L = 128, r = 0.03$ in \eqref{nosolSSCon} and assuming $\zeta = 0.9999$, yields at most $T = 1.25 \cdot 10^{231}$ attack opportunities for Steve.
}

The previous KPA analyses hinge on a counting argument in a general setting, without any other side information on the structure of $A^{(1)}$ or $\Delta A$. 
As we will show in the experiments of Section \ref{sec:ae}, KPAs yield no advantage in terms of recovery performances to unintended receivers.
Obviously, as further prior information becomes available (for example the knowledge that the unknown $\Delta A$ has additional structure, or that the original signal is distributed is a non-uniform fashion 
\cite{Mangia_2012, cambareri2013design}) revealing the hidden information may be easier. Yet, this is true for any {encryption} scheme in which either the encryption key or the plaintext have a non-uniform distribution and is out of the scope of this analysis.

\subsection{Signal Recovery-Based Class-Upgrade Attacks}
\label{othercua}
\noindent Class-upgrade attacks to two-class encryption schemes are closely related to a recovery problem setting that has attracted some attention {in prior works}, \ie, {\em sparse signal recovery under matrix uncertainty}. 
To recast our problem in this setting, we may construct such a signal recovery-based attack by letting $A^{(1)}=A^{(0)}+\Delta A$ as the encoding matrix, where $A^{(0)}$ is known {\em a priori} and $\Delta A$ is an unknown random perturbation matrix. This information is paired with the knowledge of the ciphertext $y$ and a prior on the unknown plaintext $x$, that is known to be sparse in a basis $D$. Thus, we attempt the {\em joint recovery} of ${x}$ and $\Delta A$, eventually just leading to a refinement of the estimated $\hat{x}$. Two main algorithms are capable of addressing specifically this problem setup for a generic $\Delta A$, namely Generalised Approximate Message-Passing under Matrix Uncertainty (MU-GAMP \cite{parker2011compressive}) and Sparsity-cognisant Total Least-Squares (S-TLS \cite{zhu2011sparsity}). 

Although appealing, this joint recovery approach can be anticipated to fail for multiple reasons. 
First, this attack is intrinsically harder than Steve's KPA in that the true plaintext $x$ is here unknown. Whatever $\Delta A$ is a candidate solution to Steve's KPA given $x$, is also a possible solution of joint recovery with the same $x$ as a further part of the solution. Since we know from Section \ref{cukpa} that Steve's KPA typically has a huge number of indistinguishable and equally-sparse candidate solutions, at least as many will verify the joint recovery problem when the plaintext is also unknown. 
Hence, this approach has negligible odds of yielding more information on $\Delta A$ than Steve's KPA.

{Note that this relationship between the set of solutions to Steve's KPA and joint recovery-based attacks also prevents the latter from being of any use as a {\em refinement step} to improve $\Delta A$ after its guess by an initial KPA. In fact, recovering an estimate of $x$ in this case would be to no avail, since the true $x$ must be known {\em a priori} in the initial KPA. 

Notwithstanding this, the above joint recovery approach estimates $x$ along with a new $\Delta A$; thus, the best-case achievable signal recovery {is} the true $x$, for which the candidate solutions in $\Delta A$ are at best identical to those of the initial KPA, as by \eqref{CSkpsteve} they must verify $\varepsilon = \Delta A x$. No improvement is therefore obtained by applying joint recovery after Steve's KPA.  
}

Furthermore, going back to simple joint-recovery, note that it amounts to solving $y=A^{(0)}x+\Delta A x$ with $\Delta A$ and $x$ unknown, that is clearly a non-linear equality involving non-convex/non-concave operators. In general, this is a hard problem; both the aforementioned algorithms are indeed able to effectively compensate matrix uncertainties when $\Delta A$ depends on a low-dimensional, {\em deterministic} set of parameters. However, such a model does not apply to two-class encryption: even if $\Delta A$ is $c$-sparse, it has no deterministic structure -- to make it so, one would need to know the exact set $C^{(0)}$ of $c$ index pairs at which the sign flipping randomly occurred, which by itself entails a combinatorial search. 

In fact, $\Delta A$ is \emph{uniform} in the sense of \cite{parker2011compressive} since it may be regarded as a realisation of a random matrix with i.i.d. zero-mean, bounded-variance entries (as also detailed in \cite{cambareri2014low}). 
Hence, we expect the accuracy of the estimate $\hat{x}$ with joint recovery (both using S-TLS and MU-GAMP) to agree with the uniform matrix uncertainty case of \cite{parker2011compressive}, where negligible improvement is shown with respect to the (non-joint) recovery algorithm GAMP \cite{rangan2011generalized}. The advocated reason is that the perturbation noise $\varepsilon = \Delta A x$ is asymptotically Gaussian for a given $x$ \cite[Proposition 2.1]{parker2011compressive}.

We now provide some empirical evidence on the ineffectiveness of joint recovery as a class-upgrade attack for finite $n, m$ and sparsity $k$.
\begin{figure}[t]
	\centering
	\includegraphics[width=3.4in,clip=true,trim=10 0 0 0]{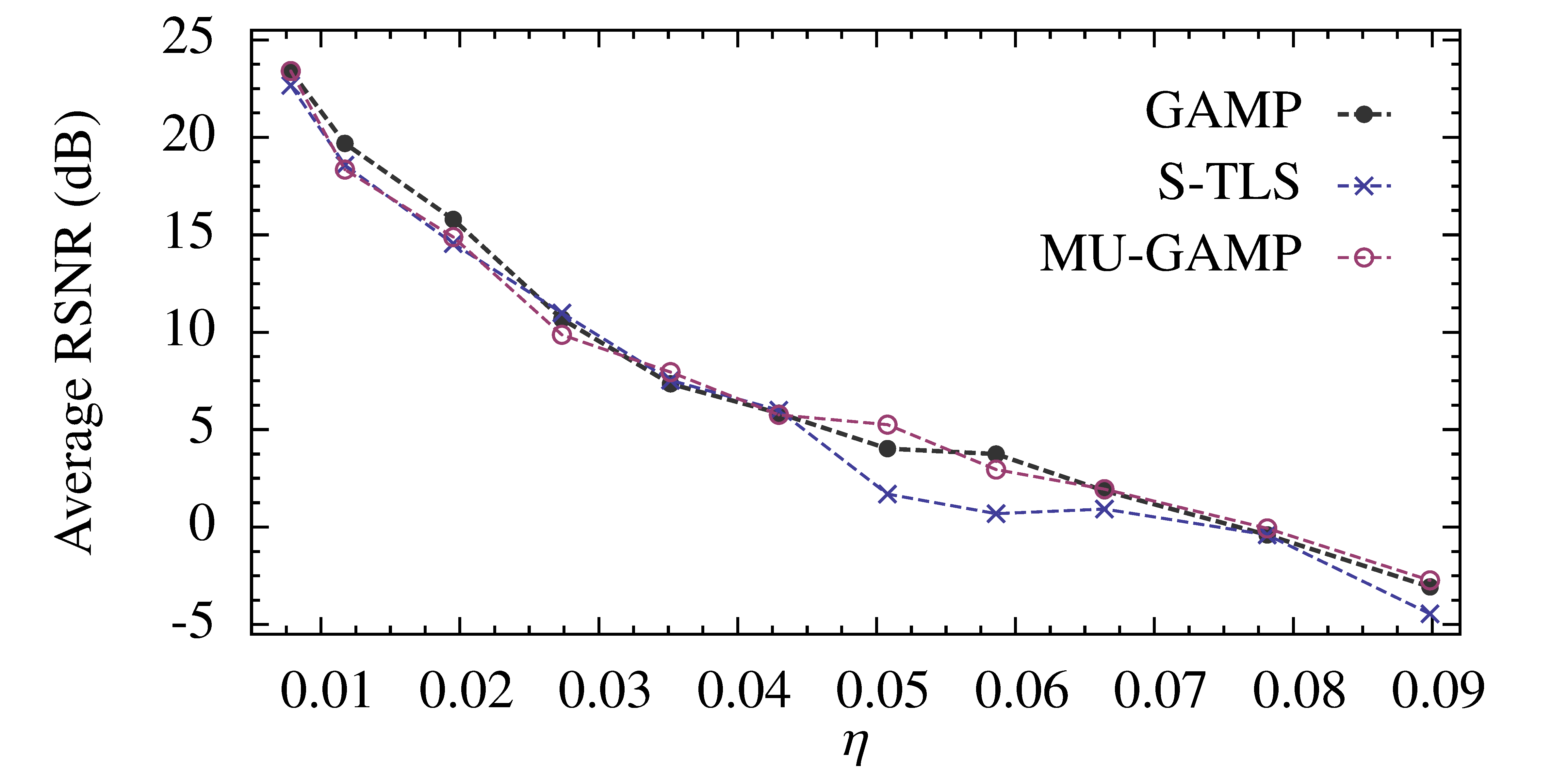}
	\caption{Average recovery signal-to-noise ratio performances of a class-upgrade attack using signal recovery under matrix uncertainty algorithms.}
	\label{fig:jointCSrnd}
\end{figure}
As an example, we let $n=256$, $m = 128$, $k = 20$ and $\eta = \frac{c}{m n} \in [0.005, 0.1]$ and generate $100$ random instances of $x = D s$ with $s$ which is $k$-sparse with respect to a randomly selected, known orthonormal basis $D$. For each $\eta$, we also generate $100$ pairs of matrices $(A^{(0)}, A^{(1)})$ related as \eqref{signflip} and encode $x$ by $y=A^{(1)}x$. 
Signal recovery is performed by MU-GAMP, S-TLS and GAMP. To maximise their performances, each of the algorithms is run with parameters provided by a ``genie'' revealing the exact value of the unknown features of $x$. In particular, MU-GAMP and GAMP are provided with an i.i.d. Bernoulli-Gaussian sparsity-enforcing signal model \cite{rangan2011generalized, vila2011expectation} having the exact mean, variance and sparsity level of the instances $s$. As far as the perturbation $\Delta A$ is concerned, MU-GAMP is given the probability distribution of its i.i.d. entries. On the other hand, GAMP is initialised with the noise variance of $\varepsilon=\Delta A x$, that is assumed Gaussian with i.i.d. entries. S-TLS is run in its locally-optimal, polynomial-time version \cite[Section IV-B]{zhu2011sparsity} and fine-tuned with respect to its regularisation parameter as $\eta$ varies. 

We here focus on measuring the Average\footnote{$\hat{\mathbb{E}}(\cdot)$ denotes the sample average over a set of trials.} Recovery Signal-to-Noise Ratio of the latter, $\unit{ARSNR\, (dB)} = 10 \log_{10} \hat{\mathbb{E}}\left(\textstyle \frac{\|x\|^2_2}{\|x - \hat{x}\|^2_2}\right)$ reported in Fig. \ref{fig:jointCSrnd}. The standard deviation from this average is less than $\unit[1.71]{dB}$ in all the reported curves. The maximum $\unit{ARSNR}$ performance gap between GAMP and MU-GAMP is $\unit[1.22]{dB}$ while S-TLS attains generally lower performances for high values of $\eta$. These observed performances confirm what is also found in \cite{parker2011compressive}, \ie, that GAMP, MU-GAMP and S-TLS substantially attain the same performances under uniform matrix uncertainty. As expected, class-upgrade attacks based on joint recovery are ineffective even for finite $n$ and $m$, since GAMP under the same conditions is the reference case adopted in \cite[Section IV]{cambareri2014low} for the design of two-class encryption schemes. 

\section{Numerical Examples}
\label{sec:ae}
\begin{figure*}[t] 
	\centering
	\null\hfill\subfloat[\label{fig:EveKPA_ECG}]{\includegraphics[height=175pt]{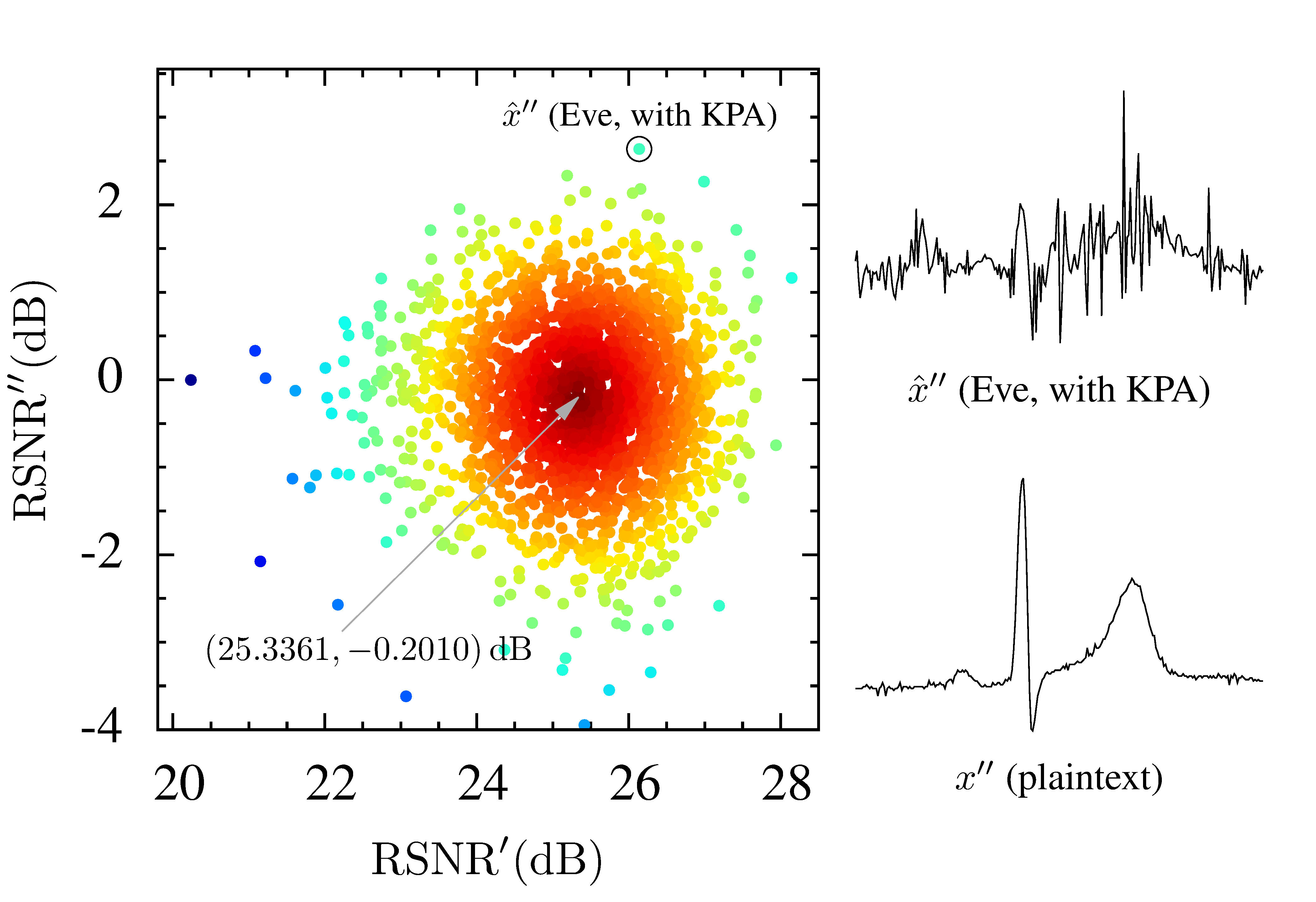}}\hfill
	\subfloat[\label{fig:SteveKPA_ECG}]{\includegraphics[height=175pt]{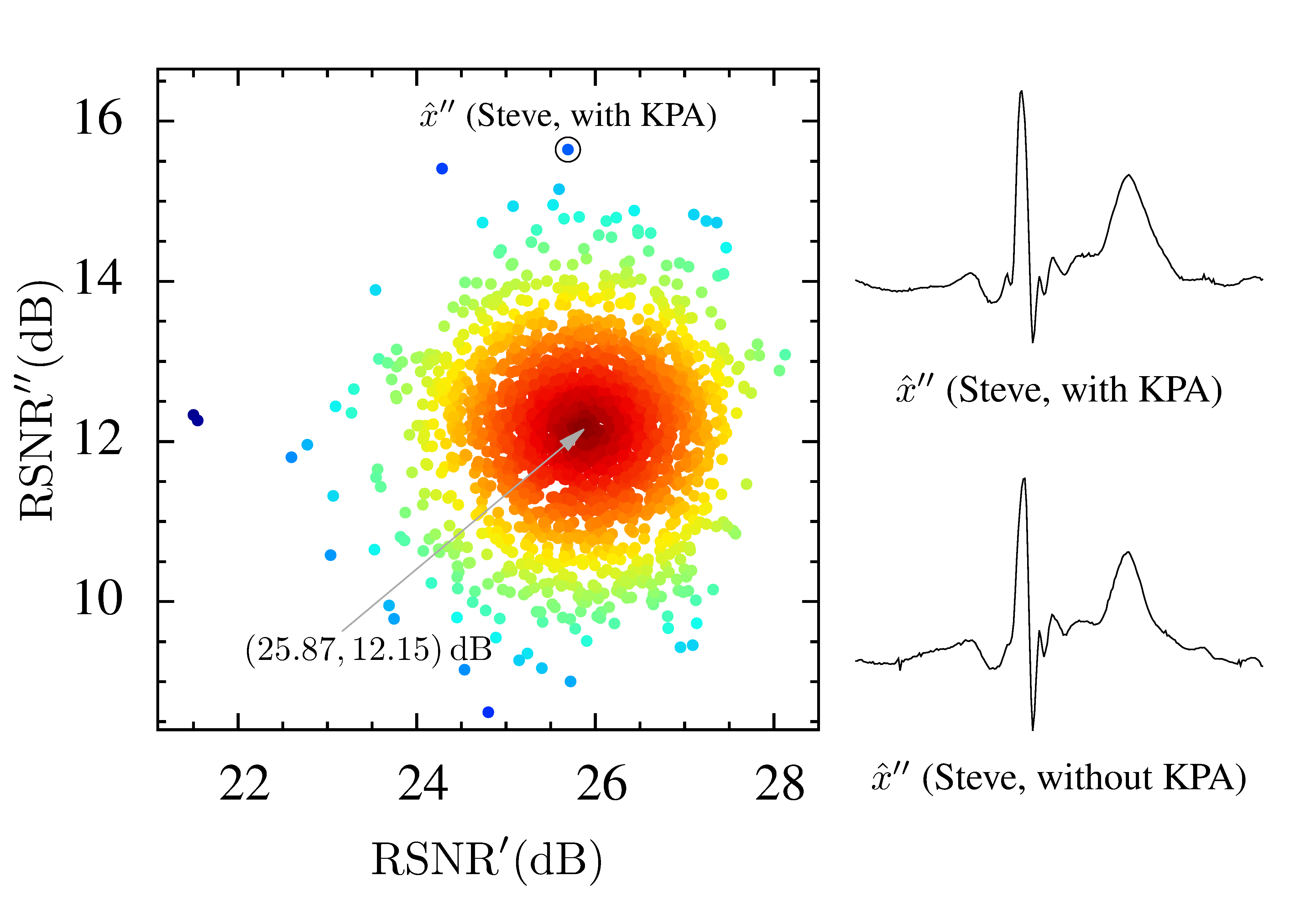}}\hfill\null
	\caption{\label{fig:ECG} Effectiveness of (a) Eve and (b) Steve's KPA in recovering a hidden ECG. Each point is a guess of the encoding matrix $A^{(1)}$ whose quality is assessed by decoding the ciphertext $y'$ corresponding to the known plaintext $x'$ ($\unit{RSNR}'$) and by decoding a new ciphertext $y''$ ($\unit{RSNR}''$). The Euclidean distance from the average $(\unit{RSNR}', \unit{RSNR}'')$ is highlighted by colour gradient.}
\end{figure*}

{This Section aims at providing an intuitive appreciation of the poor quality obtained by signal recovery with KPA solutions. While the objective of KPAs is cryptanalysing the true encoding matrix to ultimately retrieve the encryption key, we here focus on the properties of KPA solutions as encoding matrix guesses that can, in the attackers' belief, improve their signal recovery quality. Thus, we verify that this improvement does not occur by exemplifying {practical cases of} KPAs in a common framework, which follows this procedure:
\begin{enumerate}[leftmargin=*]
\setlength\itemsep{0.25em}
\item {\em Attack}: an attacker performing a KPA gains access to a single plaintext-ciphertext pair $(x', y')$, and attacks the corresponding true encoding matrix $A^{(1)}$ row-by-row; we here infer each row $A^{(1)}_j$ by generating instances of an i.i.d. antipodal random vector until a large number of candidate solutions $\hat{A}^{(1)}_j$ that verify $y'_j=\hat{A}^{(1)}_j x'$ is found. 

Thus, the inferred $\hat{A}^{(1)}$ is composed by collecting the outputs of $m$ Monte Carlo random searches for the corresponding matrix rows. This generation approach is preferable to solving each attacker's KPA by means of CPLEX's binary programming solver for two reasons. Firstly, it is known from Theorem \ref{ExpSolSS} that the expected number of solutions is very large and thus the probability of {finding one by} random search is far from being negligible, while its computational cost is relatively low. {Secondly}, the theoretical conditions \cite{candes2008restricted} that guarantee $x'$ can be retrieved from $y'$ despite the dimensionality reduction are applicable when $A^{(1)}$ is a typical realisation of an antipodal random matrix. On the contrary, integer programming solvers explore solutions in a systematic way, and tend to generate them in an ordered fashion. When only some of these solutions are considered (as obliged when $n$ is large), this {ordered approach yields} non-typical sets of $\hat{A}^{(1)}_j$ that could be very distant from $A^{(1)}_j$;

\item {\em Signal Recovery}: to test its guess $\hat{A}^{(1)}$, the attacker may then pretend to ignore the {\em known} $x'$ and recover {an} approximation $\hat{x}'$ from $(y', \hat{A}^{(1)})$ by using a high-performance signal recovery algorithm {such as GAMP \cite{rangan2011generalized}, optimally tuned as in Section \ref{othercua}}. In this setting we measure its accuracy by the Recovery Signal-to-Noise Ratio, $\unit{RSNR}'=10\log_{10}\frac{\|x'\|^2_2 }{\|x'-\hat{x}'\|^2_2}$, which is the only quality indicator in the attacker's perspective for $\hat{A}^{(1)}$. The $\unit{RSNR}'$ performances are here expected to match those of a (first-class) receiver fully informed on $A^{(1)}$, as the equality $y' = \hat{A}^{(1)} x'$ is verified regardless of the exactness of $\hat{A}^{(1)}$;

\item {\em Verification}: as a further test of $\hat{A}^{(1)}$, the attacker attempts the recovery of a second, {\em unknown} plaintext ${x}''$ encoded as $y''={A}^{(1)} {x}''$, of which it is only {\em known} that it was obtained with the same encoding matrix as $y'$. 

The recovery $\hat{x}''$ is then obtained by means of GAMP, yielding a new $\unit{RSNR}''=10\log_{10}\frac{\|x''\|^2_2}{\|x''-\hat{x}''\|^2_2}$ {\em unknown} to the attacker. If any point with high $\unit{RSNR}'' \approx \unit{RSNR}'$ is found, this will indicate the attacker's success at guessing $\hat{A}^{(1)}$ close to the true $A^{(1)}$. We will show how this never occurs with a large number of candidate solutions, and detail how the observed $(\unit{RSNR}',\unit{RSNR}'')$ pairs are distributed. 
\end{enumerate}
Both the practical examples of Eve and Steve's KPA follow the same procedure, with the exception that Eve directly generates $\hat{A}^{(1)}_j$, whereas} Steve generates each row $\hat{A}^{(1)}_j$ by random search of the index set $C^{(0)}_j$ that maps the {\em known} $A^{(0)}_j$ to the guess $\hat{A}^{(1)}_j$ that verifies $y'_j= \hat{A}^{(1)}_j x'$. Coherently with the theoretical setting of Section \ref{cukpa}, we also assume that Steve knows that exactly $c_j$ entries {of $A^{(0)}$} have been flipped {in each row of $A^{(1)}$}. Repeating this search for $m$ rows in both attacks provides Eve and Steve's candidate solutions $\hat{A}^{(1)}$, of which we will study how the corresponding $(\unit{RSNR}',\unit{RSNR}'')$ pairs are distributed as mentioned above.

\subsection{Electrocardiographic Signals}
\noindent We now consider ECG signals from the MIT PhysioNet database \cite{PhysioNet} sampled at $f_s=\unit[256]{Hz}$ and encoded as described, from two windows $x', x''$ of $n = 256$ samples (and quantised with $B_x = \unit[12]{bit}$) into the measurement vectors $y', y''$ of dimensionality $m = 90$. Decoding is allowed by the sparsity level of the windowed signal when decomposed with $D$ chosen as {a Symmlet-6 orthonormal wavelet basis} \cite{mallat1999wavelet}.

We generate $2000$ candidate solutions for both Eve and Steve's KPA that correspond to the recovery performances reported in Fig. \ref{fig:ECG}. While both malicious users are able to reconstruct the known plaintext $x'$ with a relatively high average $\unit{RSNR}' \approx \unit[25]{dB}$ (their KPAs indeed yield solutions to $y' = \hat{A}^{(1)} x'$), on the second window of samples $x''$  the eavesdropper achieves an average $\unit{RSNR}''\approx \unit[-0.20]{dB}$ (Fig. \ref{fig:EveKPA_ECG}), whereas the second-class decoder achieves an average $\unit{RSNR}''\approx \unit[12.15]{dB}$ (Fig. \ref{fig:SteveKPA_ECG}) when the two-class encryption scheme is set to a sign flipping density $\eta=\nicefrac{c}{m n} = 0.03$ between $A^{(0)}$ and $A^{(1)}$. In this case, the nominal second-class $\unit[]{RSNR} = \unit[11.08]{dB}$ when reconstructing $x''$ from $y''$ with $A^{(0)}$, while the correlation coefficient between $\unit[]{RSNR}'$ and $\unit[]{RSNR}''$ is $0.0140$; these figures clearly highlight the ineffectiveness of KPAs at inferring $A^{(1)}$ in this case. This is also confirmed by the perceptual quality of $\hat{x}''$ corresponding to the maximum $\unit{RSNR}''$ highlighted in Fig. \ref{fig:ECG}.
\subsection{Sensitive Text in Images} 
\begin{figure*}[t] 
	\centering	
 \null\hfill\subfloat[\label{fig:EveKPA_IMG}]{\includegraphics[height=175pt]{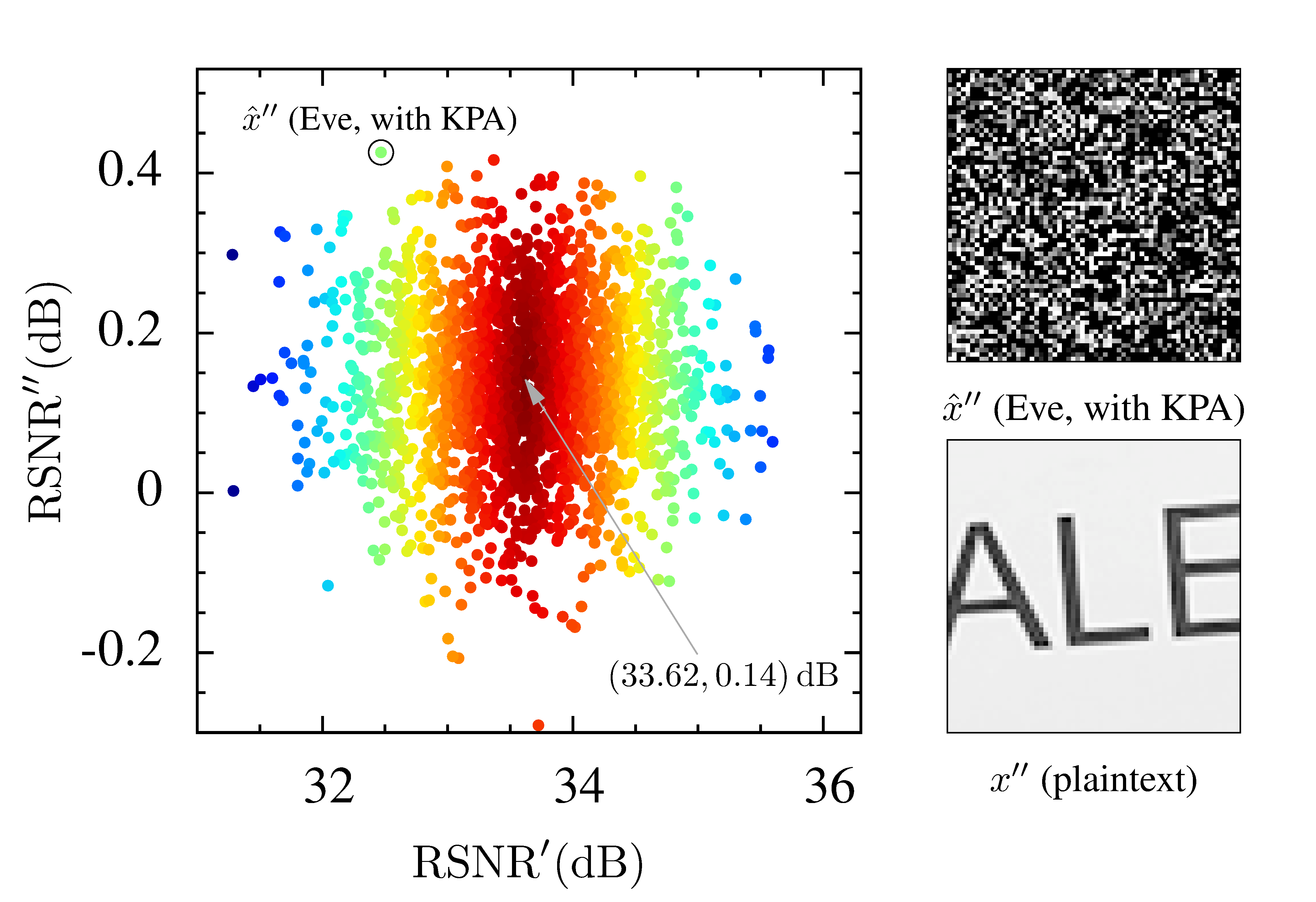}}\hfill\subfloat[\label{fig:SteveKPA_IMG}]{\includegraphics[height=175pt]{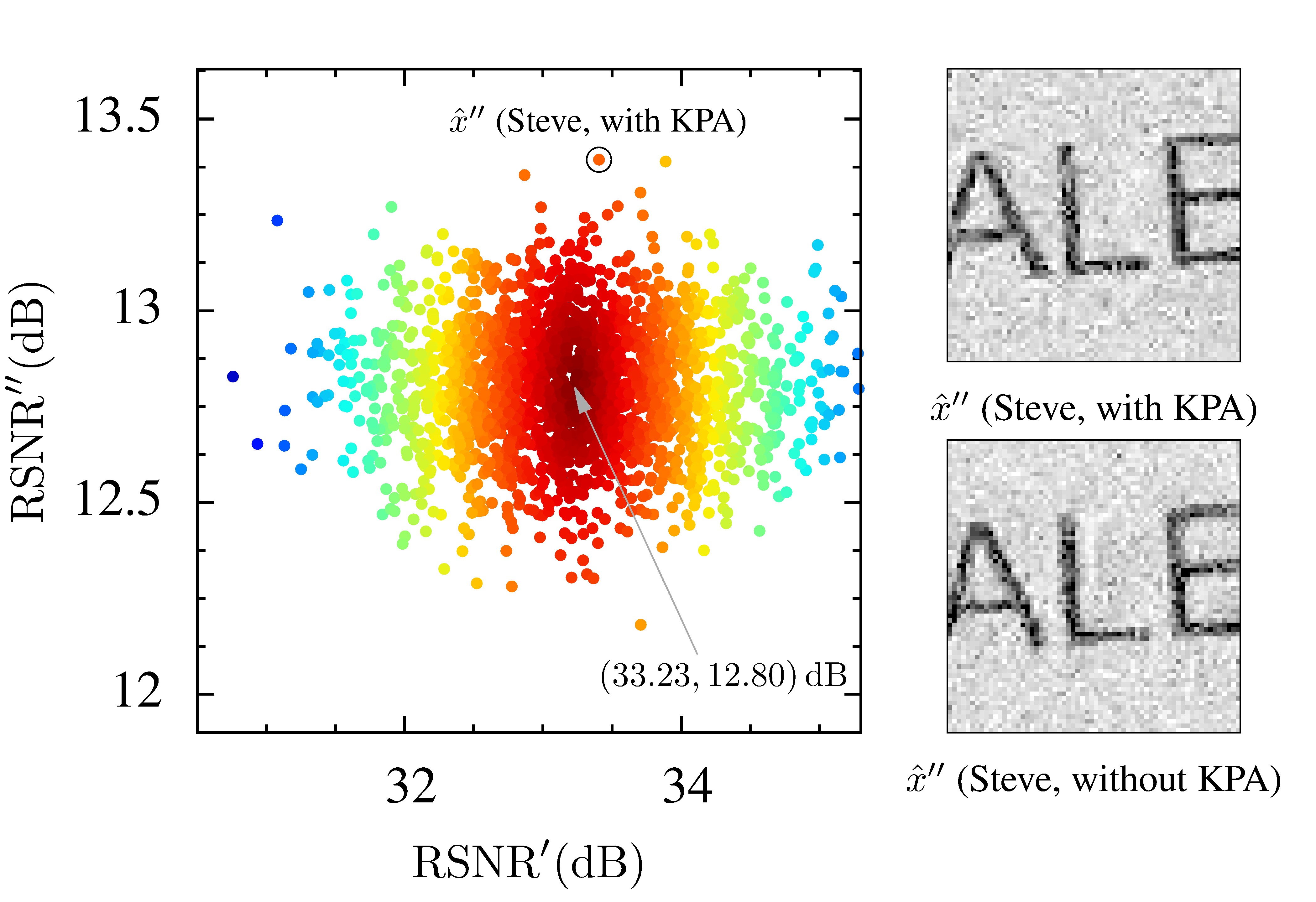}}\hfill\null
	\caption{\label{fig:Img}Effectiveness of (a) Eve and (b) Steve's KPA in recovering hidden image blocks. Each point is a guess of the encoding matrix $A^{(1)}$ whose quality is assessed by decoding the ciphertext $y'$ corresponding to the known plaintext $x'$ ($\unit{RSNR}'$) and by decoding a new ciphertext $y''$ ($\unit{RSNR}''$). The Euclidean distance from the average $(\unit{RSNR}', \unit{RSNR}'')$ is highlighted by colour gradient.}	
\end{figure*}
\noindent In this example we consider the same test images used in \cite{cambareri2014low}, \ie, $\unit[640 \times 512]{pixel}$ grayscale images of people holding a printed identification text concealed by means of two-class encryption. To reduce the computational burden of KPAs we assume a block size of $\unit[64 \times 64]{pixel}$, $B_x = \unit[8]{bit}$ per pixel, and encode the resulting $n=4096$ pixels into $m=2048$ measurements. Signal recovery is performed by assuming the blocks have a sparse representation on a 2D Daubechies-4 wavelet basis \cite{mallat1999wavelet}. Two-class encryption is applied on the blocks containing printed text: we choose two adjacent blocks $x', x''$ containing some letters and encoded with the same $A^{(1)}$; in this case, the second-class decoder nominally achieves $\unit{RSNR}=\unit[12.57]{dB}$ without attempting class-upgrade due to the flipping of $c = 251658$ entries (corresponding to a perturbation density $\eta = 0.03$) in the encoding matrix.

In order to test Eve and Steve's KPA we randomly generate $2000$ solutions for the $j$-th row of the encoding given $x',y'$: it is worth noting that while in the previous case the signal dimensionality is sufficiently small to produce a solution set in less than two minutes, in this case generating $2000$ different solutions for a single row may take up to several hours for some particularly hard instances.

By using these candidate solutions to find $\hat{x}', \hat{x}''$ we obtain the results of Figure \ref{fig:Img}: while both attackers attain an average $\unit{RSNR}'\approx \unit[33]{dB}$ on $x'$, Eve is only capable of reconstructing $x''$ with an average $\unit{RSNR}''\approx \unit[0.14]{dB}$ where Steve reaches an average $\unit{RSNR}''\approx \unit[12.80]{dB}$ with $\eta = 0.03$.
Note also that, although some lucky guesses exist with $\unit{RSNR}''>\unit[12.57]{dB}$, it is impossible to identify them by looking at $\unit{RSNR}'$ since the correlation coefficient between $\unit{RSNR}'$ and $\unit{RSNR}''$ is $-0.0041$. Therefore, Steve cannot rely on observing the $\unit{RSNR}'$ to choose the best performing solution $\hat{A}^{(1)}$, so both Eve and Steve's KPAs are inconclusive. As a further perceptual evidence of this, the best recoveries according to the $\unit{RSNR}''$ are reported in Fig. \ref{fig:Img}.

\section{Conclusion}
\noindent In this paper we have analysed known-plaintext attacks as they may be carried out on standard CS schemes with antipodal random encoding matrices as well as on the particular multiclass protocol developed in \cite{cambareri2014low}. In particular, the analysis was carried out from the two perspectives of an eavesdropper and a second-class user trying to guess the true encoding matrix. 
In both cases we have mapped multiclass CS into a collection of subset-sum problems with the aim of counting the candidate encoding matrices that match a given plaintext-ciphertext pair. In the eavesdropper case we have found that for each row the expected number grows as $O(n^{-\frac{1}{2}} \cdot 2^n)$ -- finding the true solution among such huge sets is infeasible. 
A further study of the candidate solutions' Hamming distance from the true one showed that, as the dimensionality $n$ increases, the expected number of solutions close to the true one is only a small fraction of the solution set.
As for the second-class user we have shown that depending on the available information on the true encoding matrix, the expected number of solutions is significantly smaller, yet sufficiently high for large $n$ to reassure that a second-class user will not be able to perform class-upgrade. {Moreover, other class-upgrade attacks based on signal recovery under matrix uncertainty were shown to yield almost identical performances to those of a standard decoding algorithm.}

Finally, we showed some simulated cases of KPAs on real-world signals such as ECG traces and images by running a random search for a solution set corresponding to realistic plaintext-ciphertext pairs, and afterwards tested whether any of the returned candidate solutions could lead to finding the true encoding matrix by testing them on a successive ciphertext.
In all the observed cases, we have found that the decoding performances match the average $\unit{RSNR}$ level prescribed by the multiclass encryption protocol, \ie, both malicious users are unable to successfully decode other plaintexts with significant and stable quality improvements with respect to their available prior information.

\appendices
\section{Proofs on Eavesdropper's KPA}
\label{appSS}

\noindent The following definition is used in Appendices A and C.
\begin{defn}
We define the functions
\begin{eqnarray}
F_p(a,b) & = & \int_0^1 \frac{\xi^p}{1+e^{a\xi-b}} \dd\xi \label{fp}\\
G_p(a,b) & = & \int_0^1 \frac{\xi^p}{\left(1+e^{a\xi-b}\right)\left(1+e^{b-a\xi}\right)} \dd\xi \label{gp}
\end{eqnarray}
\end{defn}

\begin{proof}[Proof of Proposition 1]
Define the binary variables $b_l \in \{0, 1\}$ so that $\sign{x_l} \hat{A}^{(1)}_{j,l} = 2 b_l - 1$ and the positive coefficients $u_l = |x_l|$. With this choice \eqref{eq:pm} is equivalent to $y_j = \sum_{l=0}^{n-1} (2 b_l - 1) u_l $ which leads to a SSP with $\upsilon =  \frac{1}{2}\left(y_j + \sum^{n-1}_{l=0} |x_l|\right)$.
Since we know that each measurement $y_j$ must correspond to the inner product between $x$ and the row $A^{(1)}_{j}$, the latter's entries are straightforwardly mapped to the {\em true} solution of this SSP, $\{\bar{b}_l\}_{l=0}^{n-1}$. 
\end{proof}

\begin{proof}[Proof of Theorem \ref{ExpSolSS}]
Let us first note that, for large $n$, $\upsilon$ in Proposition \ref{prop:SS1} is an integer in the range $\{0, \ldots, \nicefrac{nL}{2}\}$, with the values outside this interval being asymptotically unachievable as $n \rightarrow \infty$ (see \cite[Section 4]{sasamoto2001statistical}). We let $\tau =\nicefrac{\upsilon}{nL}$, $\tau \in [0, \nicefrac{1}{2}]$, and $a(\tau)$ be the solution in $a$ of the equation $\tau=F_1(a,0)$ (\ie\, \cite[(4.2)]{sasamoto2001statistical}) that is unique since $F_p(a,0)$ in \eqref{fp} is monotonically decreasing in $a$. 

From \cite[(4.1)]{sasamoto2001statistical} the number of solutions of a SSP with integer coefficients $\{u_l\}^{n-1}_{l=0}$ uniformly distributed in $\{1, \ldots, L\}$ is
\[
\SEve(\tau,n,L)\asimeq
\dfrac{
e^{
n\left[a\left(\tau\right)\tau+\int_0^1 
\log\left(1+e^{-a\left(\tau\right) \xi}\right)
\dd\xi
\right]
}
}
{
\displaystyle
\sqrt{
2\pi n L^2
G_2(a(\tau),0)
}
}
\] 
that we anticipate to have an approximately Gaussian profile (see Fig. \ref{fig:approx_gauss}). 
\begin{figure}
	\centering
		\includegraphics[height=160pt]{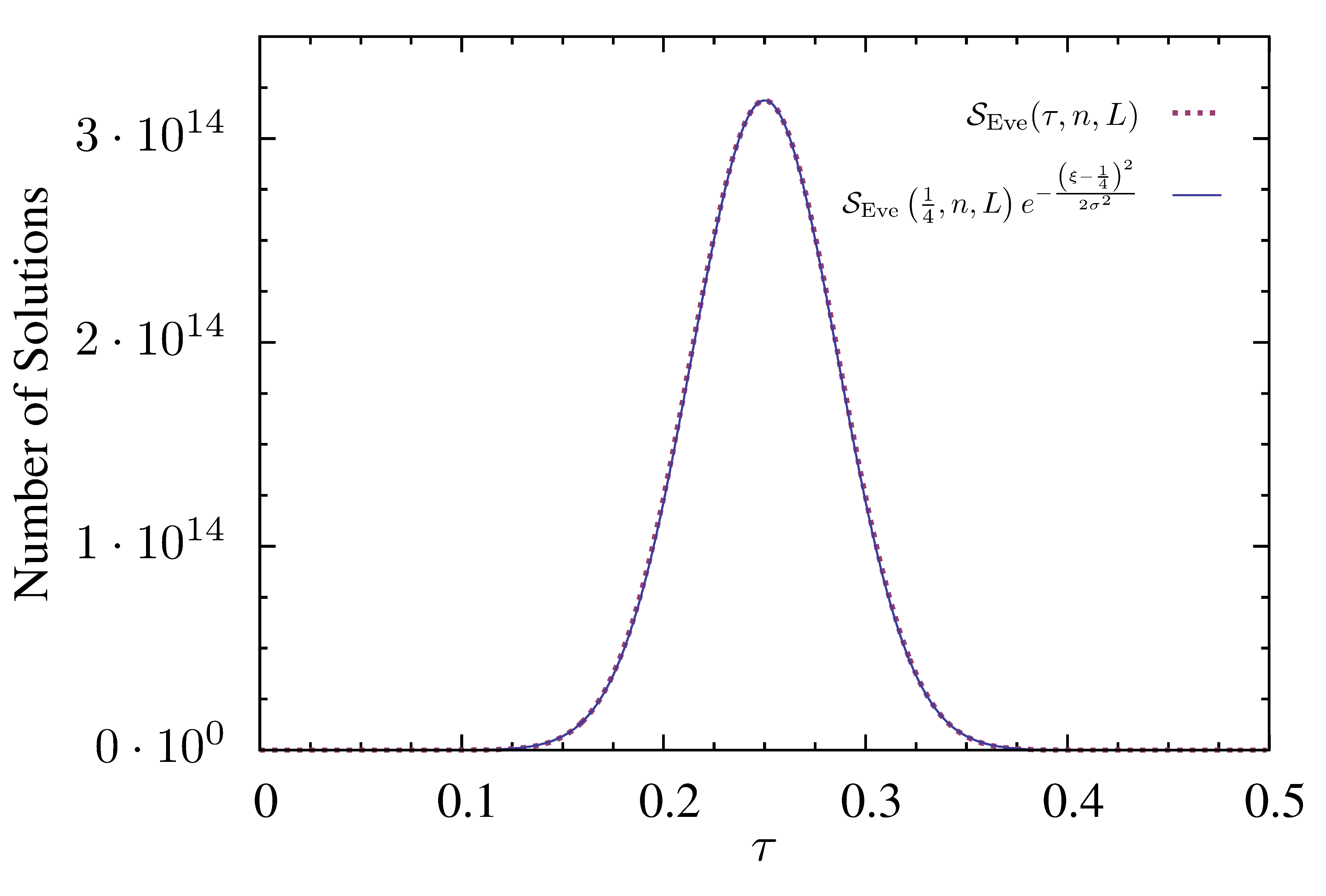}
	\caption{Gaussian approximation of $\SEve(\tau,n,L)$ for $n = 64, L = 10^4$ by letting $\sigma^2 \approx \nicefrac{1}{12 n}$.}
	\label{fig:approx_gauss}
\end{figure}
We now compute the average of $\SEve(\tau,n,L)$ in $\tau$, that clearly depends on the probability of selecting any value of $\upsilon \in \{0, \ldots, \frac{nL}{2}\}$, \ie, of $\tau \in [0, \frac{1}{2}]$. Since it is the result of a linear combination, the probability that a specific value of $\upsilon$ appears in a random instance of the SSP is proportional to the number of solutions associated to it. In normalised terms, the PDF of $\tau$ must be proportional to $\SEve(\tau,n,L)$, \ie, $\tau$ is distributed as
\[
f_\tau(t)=
\frac{1}{\int_0^{\frac{1}{2}}\SEve(\xi,n,L)\dd\xi}
\begin{cases}
\displaystyle \SEve(t,n,L),
 & \text{$0 \le t\le \frac{1}{2}$
}\\
0, & \text{otherwise}
\end{cases}
\]
\noindent With $f_\tau(t)$ we can compute the expected number of solutions:
\begin{equation}
\label{eq:aveS}
\Exp_\tau[\SEve(\tau,n,L)]=\frac{
\displaystyle
\int_0^{\frac{1}{2}}
\SEve^2(\xi,n,L)\dd\xi
}{
\displaystyle
\int_0^{\frac{1}{2}}
\SEve(\xi,n,L)\dd\xi
}
\end{equation}
Although we could resort to numerical integration, \eqref{eq:aveS} can be simplified by exploiting what noted above, \ie, that $\SEve(\tau,n,L)$ has an approximately Gaussian profile in $\tau$ (Fig. \ref{fig:approx_gauss}) with a maximum in $\tau=\nicefrac{1}{4}$. 
Hence, the expectation in $\tau$ becomes 
\begin{IEEEeqnarray}{c}
\Exp_\tau[\SEve(\tau,n,L)]\asimeq
\SEve\left(\frac{1}{4},n,L\right)\dfrac{\displaystyle\int_{-\infty}^{\infty}\left(e^{-\frac{\left({\xi-\frac{1}{4}}\right)^2}{2 \sigma ^2}}\right)^2\dd\xi}{\displaystyle\int_{-\infty}^{\infty}e^{-\frac{\left({\xi-\frac{1}{4}}\right)^2}{2\sigma^2}}\dd\xi}\nonumber \\ 
= \SEve\left(\frac{1}{4},n,L\right)\frac{1}{\sqrt{2}}= \frac{2^n}{L}\sqrt{\frac{3}{\pi n}}
\end{IEEEeqnarray}
that is actually independent of the $\sigma^2$ used in the Gaussian approximation, and in which we have exploited $a(\nicefrac{1}{4})=0$ to obtain the statement of the theorem.
\end{proof}

\section{Hamming Distance of KPA Solutions}
\label{HDsols}
\begin{proof}[Proof of Theorem \ref{ExpSolSSHam}]
\noindent We here concentrate on counting the number of candidate solutions $\{b_l\}_{l=0}^{n-1}$ to Eve's KPA that differ from the true one, $\{\bar{b}_l\}_{l=0}^{n-1}$, by exactly $h$ components (at Hamming distance $h$). We assume that $K\subseteq\{0,\dots,n-1\}$ is the set of indexes for which there is a disagreement, \ie, for all $l\in K$ we have $b_l=1-\bar{b}_l$; this set has cardinality $h$, and is one among $\binom{n}{h}$ possible sets. Since both 
$\{b_l\}_{l=0}^{n-1}$ and $\{\bar{b}_l\}_{l=0}^{n-1}$ are solutions to the same SSP, and that $b_l=\bar{b}_l$ are identical for $l\notin K$, $\sum_{l\in K}\left(1-\bar{b}_l\right) u_l=\sum_{l\in K}\bar{b}_l u_l$ must hold, implying the equality
\begin{equation}
\label{eq:part}
\sum_{\substack{l\in K\\\bar{b}_l=0}}u_l-\sum_{\substack{l\in K\\\bar{b}_l=1}}u_l=0
\end{equation}
Although \eqref{eq:part} recalls the well-known partition problem, in our case $K$ is chosen by each problem instance that sets all $u_l$ and $\bar{b}_l$. Thus, \eqref{eq:part} holds in a number of cases that depends on how many of the $2^h L^h$ possible assignments of all $u_l$ and $\bar{b}_l$ satisfy it. The only feasible cases are for $h>1$, and to analyse them we assume $K=\{0,\dots,h-1\}$ (the disagreements occur in the first $h$ ordered indexes) without loss of generality.

Moreover, when \eqref{eq:part} holds for some $\{\bar{b}_l\}_{l=0}^{n-1}$ it also holds for $\{1-\bar{b}_l\}_{l=0}^{n-1}$. Hence, we may count the configurations that verify \eqref{eq:part} with $\bar{b}_0=0$, knowing that their number will be only {\em half} of the total. With this, the configurations with $\bar{b}_0=0$ must have $\bar{b}_l=1$ for {\em at least} one $l>0$ in order to satisfy \eqref{eq:part}, giving $2^{h-1}-1$ total cases to check. 

The following paragraphs illustrate that, for $h<L$, the number of configurations that verify \eqref{eq:part} can be written as a polynomial of order $h-1$. With this in mind we can start with the explicit computation for $h=\{2, 3\}$. For $h=2$, there is only one feasible assignment for the $\{\bar{b}_l\}_{l=0}^{n-1}$, so $u_0=u_1$ in \eqref{eq:part}, which makes $2L$ cases out of $2^2 L^2$.
For $h=3$, one has $3$ feasible assignments for the $\{\bar{b}_l\}_{l=0}^{n-1}$. Due to the symmetry of \eqref{eq:part} all the configurations have the same behaviour and we may focus on, \eg, $\bar{b}_0=\bar{b}_1=0$ and $\bar{b}_2=1 \Rightarrow u_0+u_1=u_2$; this can be satisfied only when $u_0+u_1\le L$, \ie, for $\frac{L(L-1)}{2}$ configurations. This makes a total of $2 \cdot 3 \cdot  \frac{L(L-1)}{2} = 3L(L-1)$ over the $2^3 L^3$ possible configurations.

For $h>3$, this procedure is much less intuitive; nevertheless, we can at least prove that the function $P_h(L)$ counting the configurations for which \eqref{eq:part} holds is a polynomial in $L$ of degree $h-1$. To show this, let us proceed in three steps.
\begin{enumerate}
\item Indicate with $\pi_{\bar{b}}$ the $(h-1)$-dimensional subspace of $\mathbb{R}^h$ defined by 
$\sum_{\substack{l\in K\\\bar{b}_l=0}}\xi_l-\sum_{\substack{l\in K\\\bar{b}_l=1}}\xi_l=0, \xi \in \mathbb{R}^h$.
The intersection $\alpha_{\bar{b}}(L)=\{1,\ldots, L\}^h\cap\pi_{\bar{b}}$ is such that each assignment of $\{u_l\}_{l=0}^{h-1} \in \{1,\ldots, L\}^h$ satisfying \eqref{eq:part} is an integer point in $\alpha_{\bar{b}}$. To count those points define $\beta_{\bar{b}}(L)=\{0,\ldots, L+1\}\cap\pi_{\bar{b}}$ and note that the number of integer points in $\alpha_{\bar{b}}$ is equal to the number of integer points in the interior of $\beta_{\bar{b}}$ (the points on the frontier of $\beta_{\bar{b}}$ have at least one coordinate that is either $0$ or $L+1$).

Note how $\{0,\ldots,L+1\}^h$ scales linearly with $L+1$ while $\pi_{\bar{b}}$ is a subspace and therefore scale-invariant. Hence, their intersection $\beta_{\bar{b}}(L)$ is an $h-1$-dimensional polytope that scales proportionally to the integer $L+1$, as required by Ehrhart's theorem \cite{ehrhart1967probleme}. The number $E_{\bar{b}}(L)$ of integer points in $\beta_{\bar{b}}(L)$ is then a polynomial in $L+1$ (and so $L$) of degree equal to the dimensionality of $\beta_{\bar{b}}(L)$, \ie, $h-1$. From Ehrhart-Macdonald's reciprocity theorem \cite{macdonald1971polynomials} we know that the number of integer points in the interior of $\beta_{\bar{b}}$ and thus in $\alpha_{\bar{b}}$ is $(-1)^{h-1}E_{\bar{b}}(-L)$, that is also a polynomial in $L$ of degree $h-1$.

\item If two different assignments $\{\bar{b}'_l\}_{l=0}^{h-1}$ and $\{\bar{b}''_l\}_{l=0}^{h-1}$ are considered, then $\alpha_{\bar{b}'}(L)\cap \alpha_{\bar{b}''}(L)=\{1,\ldots,L\}^h\cap\pi_{\bar{b}'}\cap\pi_{\bar{b}''}$. The same argument we used above tells us that the number of integer points in such an intersection is a polynomial in $L$ of degree $h-2$ and, in general that the number of integer points in the intersection of any number of polytopes $\alpha_{\bar{b}}(L)$ is a polynomial of degree not larger than $h-1$.

\item The number of configurations of $\{u_l\}_{l=0}^{h-1}$ and $\{\bar{b}_l\}_{l=0}^{h-1}$ that satisfy \eqref{eq:part} with respect to the above $K$ is the number of integer points in the union of all possible polytopes $\alpha_{\bar{b}}$, \ie,  $\bigcup_{\{\bar{b}_l\}_{l=0}^{h-1}}\alpha_{\bar{b}}(L)$.
Such a number can be computed by the inclusion-exclusion principle that amounts to properly summing and subtracting the number of integer points in those polytopes and their various intersections.
Since sum and subtraction of polynomials yield polynomials of non-increasing degree, we know that number is the evaluation of a polynomial $P_h(L)$ with degree not greater than $h-1$.
\end{enumerate}
Let us then write $P_h(L)=\sum_{j=0}^{h-1}p^h_j L^j$. In order to compute its coefficients $p^h_j$ we may fix a binary configuration $\{\overline{b}_l\}^{h-1}_{l=0}$, count the points $\{u_l\}_{l=0}^{h-1} \in \mathbb{N}_+^h$ for which \eqref{eq:part} is verified by means of integer partition functions (that also have a polynomial expansion), and subtract the points in which $\{u_l\}_{l=0}^{h-1} \notin \{1,\ldots,L\}^h$. By summation over all binary configurations, one can extract the coefficients associated with $L^j$ for each $h$. Table \ref{tab:PhL} reports the result of this procedure as carried out by symbolic computation for $h \leq 15$.
\end{proof}
\section{Proofs on the Class-Upgrade KPA}
\label{appSScon}
\begin{proof}[Proof of Proposition 2] 
{\noindent In this case the attacker knows $(A^{(0)}, x, y)$, and is able to calculate $\varepsilon_j = y_j - \sum^{n-1}_{l=0} A^{(0)}_{j,l} x_l = \sum^{n-1}_{l=0} \Delta A_{j, l} x_l$ where the $\Delta A_{j, l}$ are unknown. For the $j$-th row, the attacker also knows there are $c_j$ non-zero elements in $\Delta A_{j,l} = - 2 A^{(0)}_{j,l} b_l$ with $b_l \in \{0, 1\}$ binary variables that are $1$ if the flipping occurred and $0$ otherwise. Note that from the above information $c_j = \sum^{n-1}_{l=0} b_l$.}
With this we define a set of even weights $D_l = -2 A^{(0)}_{j,l} x_l \in \{-2 L, \ldots,-2,0,2,\ldots, 2 L\}$ so the KPA is defined by satisfying the equalities
\begin{eqnarray}
\label{eq:SSPdpm}
\varepsilon_j & = & \sum^{n-1}_{l = 0} D_l b_l \\
\label{eq:SSPconstrf}
c_j & = & \sum^{n-1}_{l=0} b_l 
\end{eqnarray}
To obtain a standard $\gamma$-SSP with positive weights {and $\gamma = c_j$} we sum $2 L$ to all $D_l$ so \eqref{eq:SSPdpm} becomes $\varepsilon_j + 2  L\sum^{n-1}_{l=0} b_l =\sum^{n-1}_{l=0} (D_l + 2 L) b_l$. Multiplying both sides by $\nicefrac{1}{2}$ and using \eqref{eq:SSPconstrf} yields $\upsilon =  \frac{1}{2} \varepsilon_j  + L c_j = \sum^{n-1}_{l = 0} u_l b_l$ where $u_l=-A^{(0)}_{j,l} x_l +L \in \{0, \ldots, Q\}$. $Q = 2L$. Finally, we exclude $u_l = 0$ to facilitate the attack.
\end{proof}

\begin{proof}[Proof of Theorem \ref{ExpSolSSCon}]
Assume $F_p(a,b)$ and $G_p(a,b)$ as in \eqref{fp},\eqref{gp}. Define the normalised constraint $r=\frac{c_j}{n}$ and two quantities $a(\tau, r)$ and $b(\tau,r)$ that are the solutions of the following system of equalities
\begin{IEEEeqnarray}{c}
r = F_0(a,b) \nonumber \\
\tau = F_1(a,b) \nonumber
\end{IEEEeqnarray}
\noindent that are respectively equivalent to \cite[(5.3-4)]{sasamoto2001statistical}. We also define
\begin{IEEEeqnarray}{c}
\mathcal{G}(\tau, r) = \begin{pmatrix} G_0(a\left(\tau,r\right),b\left(\tau,r\right)) & G_1(a\left(\tau,r\right),b\left(\tau,r\right) \\
G_1(a\left(\tau,r\right),b\left(\tau,r\right) & G_2(a\left(\tau,r\right),b\left(\tau,r\right)) \end{pmatrix} \nonumber 
\end{IEEEeqnarray}
With this, \cite[(5.8-9)]{sasamoto2001statistical} prove that the number of solutions of a $\gamma$-SSP with integer coefficients $\{u_l\}^{n-1}_{l=0}$ uniformly distributed in $\{1, \ldots, Q\}, Q = 2 L, \gamma = c_j$ is
\begin{eqnarray}
\label{eq:eqsystem}
\lefteqn{
\SSteve(\tau,n,L,r) = \frac{
e^{n\left(
a(\tau,r)\tau-b(\tau,r)r\right)}
}
{4 \pi n L 
\sqrt{\det\left(\mathcal{G}(\tau, r)\right)}
}\cdot}\\
\nonumber
& \hspace{30mm} \cdot \, e^{n
\displaystyle\int_0^1
\log\left[1+e^{b\left(\tau,r\right)-a\left(\tau,r\right) \xi}\right]
\dd\xi
}
\end{eqnarray}
\noindent Using the same arguments as in the proof of Theorem \ref{ExpSolSS}, we average on $\tau$ and obtain an expression identical to \eqref{eq:aveS} for the computation of $\Exp_\tau[\SSteve(\tau,n,L,r)]$. 
Since $\SSteve(\tau,n,L,r)$ has once again an approximately Gaussian profile in $\tau$ with a maximum in $\tau=\frac{r}{2}$ we approximate the expectation in $\tau$, 
\begin{IEEEeqnarray}{c}
\label{eq:X2}
\Exp_\tau[\SSteve(\tau,n,L,r)]\asimeq \SSteve\left(\frac{r}{2},n,L,r\right)\frac{1}{\sqrt{2}} \nonumber \\ =
\sqrt{
\frac{3}{2}
} \frac{
r^{-1-n\rho}\left(1-r\right)^{-1-n(1-r)}
}
{
2 \pi n L
}
\end{IEEEeqnarray}
by using the fact that $a\left(\frac{r}{2},r\right)=0$ and $b\left(\frac{r}{2},r\right)=\log\left(\frac{r}{1-r}\right)$.

\end{proof}

\ifCLASSOPTIONcaptionsoff
  \newpage
\fi

\bibliographystyle{IEEEtran}
\bibliography{CompressiveSensingSecure}

\begin{IEEEbiography}
[{\includegraphics[width=1in,height=1.25in,clip,keepaspectratio]{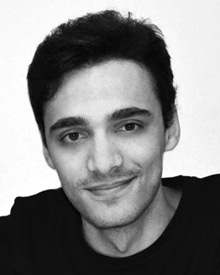}}]{Valerio Cambareri}
(S'13) received the B.S., M.S. (\emph{summa cum laude}), and Ph.D. degrees in Electronic Engineering from the University of Bologna, Italy, in 2008, 2011 and 2015 respectively. From 2012 to 2015, he was a Ph.D. student in Electronics, Telecommunications and Information Technologies at DEI -- University of Bologna, Italy. In 2014 he was a visiting Ph.D. student in the Integrated Imagers team at IMEC, Belgium. His current research activity focuses on statistical and digital signal processing, compressed sensing and computational imaging.
\end{IEEEbiography}%
% \vspace{-1.0cm}%
\begin{IEEEbiography}%
[{\includegraphics[width=1in,height=1.25in,clip,keepaspectratio]{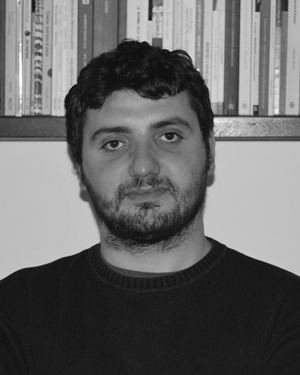}}]
{Mauro Mangia}
(S'09-M'13) received the B.S. and M.S. degree in Electronic Engineering from the University of Bologna, Italy, in 2004 and 2009 respectively; he received the Ph.D. degree in Information Technology from the University of Bologna in 2013. He is currently a post-doc researcher in the statistical signal processing group of ARCES -- University of Bologna, Italy. In 2009 and 2012 he was a visiting Ph.D. student at the \'Ecole Polytechnique F\'ed\'erale de Lausanne (EPFL). 
His research interests are in nonlinear systems, compressed sensing, ultra-wideband systems and system biology. 
He was recipient of the 2013 IEEE CAS Society Guillemin-Cauer Award and the best student paper award at IEEE ISCAS2011.
\end{IEEEbiography}%
% \vspace{-1.0cm}%
\begin{IEEEbiography}%
[{\includegraphics[width=1in,height=1.25in,clip,keepaspectratio]{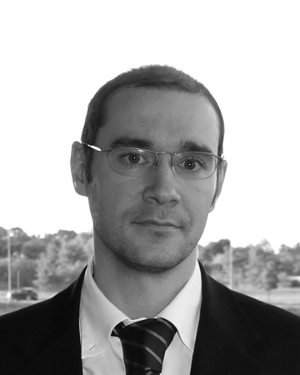}}]
{Fabio Pareschi}
(S'05-M'08) received the Dr. Eng. degree (with honours) in Electronic Engineering from University of Ferrara, Italy, in 2001, and the Ph.D. in Information Technology under the European Doctorate Project (EDITH) from University of Bologna, Italy, in 2007. He is currently an Assistant Professor in the Department of Engineering (ENDIF), University of Ferrara. He is also a faculty member of ARCES -- University of Bologna, Italy. 
He served as Associate Editor for the IEEE Transactions on Circuits and Systems -- Part II (2010-2013). His research activity focuses on analog and mixed-mode electronic circuit design, statistical signal processing, random number generation and testing, and electromagnetic compatibility. He was recipient of the best paper award at ECCTD2005 and the best student paper award at EMCZurich2005.
\end{IEEEbiography}%
% \vspace{-1.0cm}%
\begin{IEEEbiography}
[{\includegraphics[width=1in,height=1.25in,clip,keepaspectratio]{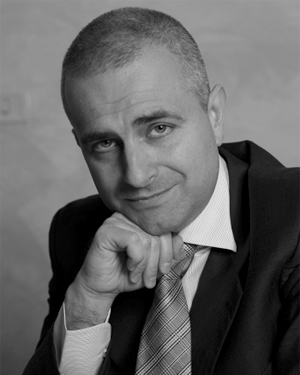}}]
{Riccardo Rovatti}
(M'99-SM'02-F'12) received the M.S. degree in Electronic Engineering and the Ph.D. degree in Electronics, Computer Science, and Telecommunications both from the University of Bologna, Italy in 1992 and 1996, respectively. 
He is now a Full Professor of Electronics at the University of Bologna. 
He is the author of approximately 300 technical contributions to international conferences and journals, and of two volumes. 
His research focuses on mathematical and applicative aspects of statistical signal processing and on the application of statistics to nonlinear dynamical systems. 
He received the 2004 IEEE CAS Society Darlington Award, the 2013 IEEE CAS Society Guillemin-Cauer Award, as well as the best paper award at ECCTD 2005, and the best student paper award at EMC Zurich 2005 and ISCAS 2011. 
He was elected IEEE Fellow in 2012 for contributions to nonlinear and statistical signal processing applied to electronic systems.
\end{IEEEbiography}%
% \vspace{-1.0cm}%
\begin{IEEEbiography}
[{\includegraphics[width=1in,height=1.25in,clip,keepaspectratio]{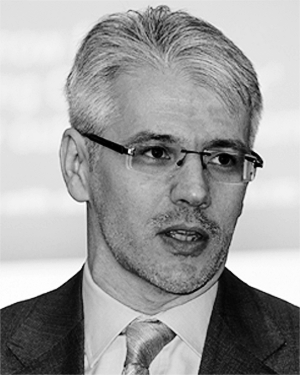}}]
{Gianluca Setti}
(S'89-M'91-SM'02-F'06) received the Ph.D. degree in Electronic Engineering and Computer Science from the University of Bologna in 1997. 
Since 1997 he has been with the School of Engineering at the University of Ferrara, Italy, where he is currently a Professor of Circuit Theory and Analog Electronics and is also a permanent faculty member of ARCES -- University of Bologna, Italy. 
His research interests include nonlinear circuits, implementation and application of chaotic circuits and systems, electromagnetic compatibility, statistical signal processing and biomedical circuits and systems. 
Dr. Setti received the 2013 IEEE CAS Society Meritorious Service Award and co-recipient of the 2004 IEEE CAS Society Darlington Award, of the 2013 IEEE CAS Society Guillemin-Cauer Award, as well as of the best paper award at ECCTD 2005, and the best student paper award at EMC Zurich 2005 and at ISCAS 2011. 
He held several editorial positions and served, in particular, as the Editor-in-Chief for the IEEE TRANSACTIONS ON CIRCUITS AND SYSTEMS -- PART II (2006-2007) and of the IEEE TRANSACTIONS ON CIRCUITS AND SYSTEMS -- PART I (2008-2009). 
Dr. Setti was the Technical Program Co-Chair at ISCAS 2007, ISCAS 2008, ICECS 2012, BioCAS 2013 as well as the General Co-Chair of NOLTA 2006. 
He was a member of the Board of Governors of the IEEE CAS Society (2005-2008), served as its 2010 President, and he is a Distinguished Lecturer of CASS (2015-2016).
He held several other volunteer positions for the IEEE and in 2013-2014 he was the first non North-American Vice President of the IEEE for Publication Services and Products.
\end{IEEEbiography}
\end{document}